\documentclass[twocolumn,10pt,fleqn]{autart}    
\usepackage[a4paper,margin={1.53 cm,1.53 cm,1.53 cm,1.53 cm}]{geometry}

\usepackage{graphicx}          
\usepackage{epsfig}
\usepackage{subfigure}
\usepackage[fleqn]{amsmath}
\usepackage{cite}
\usepackage{enumerate}
\usepackage{color}

\usepackage{algorithm}
\usepackage{algorithmic}

\newtheorem{definition}{\bf Definition}[section]

\newtheorem{remark}{\bf Remark}[section]
\newtheorem{lemma}{\bf Lemma}[section]
\newtheorem{theorem}{\bf Theorem}[section]
\newtheorem{corollary}{\bf Corollary}[section]
\newtheorem{proof}{\bf Proof}

\usepackage{amssymb}
\usepackage{amsfonts} 
\usepackage{booktabs}
\usepackage{tabularx}
\usepackage{makecell}
\usepackage{multirow}
\usepackage{indentfirst}
\usepackage{subfigure}
\usepackage{cuted}

\begin{document}

\begin{frontmatter}

\title{Locally Differentially Private Multi-Sensor Fusion Estimation With System Intrinsic Randomness
\thanksref{footnoteinfo}} 

\thanks[footnoteinfo]{\emph{Corresponding author: Hailong Huang.}}

\author[polyu]{Xinhao Yan} \ead{xin-hao-shawn.yan@connect.polyu.hk},
\author[zjut]{Bo Chen} \ead{bchen@aliyun.com},
\author[polyu]{Hailong Huang} \ead{hailong.huang@polyu.edu.hk}

\address[polyu]{Department of Aeronautical and Aviation Engineering, The Hong Kong Polytechnic University, Kowloon, Hong Kong}
\address[zjut]{Department of Automation, Zhejiang University of Technology, Hangzhou 310023, PR China}

\begin{abstract}
This paper focuses on the privacy-preserving multi-sensor fusion estimation (MSFE) problem with differential privacy considerations.
Most existing research efforts are directed towards the exploration of traditional differential privacy, also referred to as centralized differential privacy (CDP).
It is important to note that CDP is tailored to protect the privacy of statistical data at fusion center such as averages and sums rather than individual data at sensors, which renders it inappropriate for MSFE.
Additionally, the definitions and assumptions of CDP are primarily applicable for large-scale systems that require statistical results mentioned above.
Therefore, to address these limitations, this paper introduces a more recent advancement known as \emph{local differential privacy (LDP)} to enhance the privacy of MSFE.
We provide some rigorous definitions about LDP based on the intrinsic properties of MSFE rather than directly presenting the assumptions under CDP.
Subsequently, the LDP is proved to be realized with system intrinsic randomness, which is useful and has never been considered before.
Furthermore, the Gaussian mechanism is designed when the intrinsic randomness is insufficient.
The lower bound of the covariance for extra injected Gaussian noises is determined by integrating system information with privacy budgets.
Moreover, the optimal fusion estimators under intrinsic and extra disturbances are respectively designed in the linear minimum variance sense.
Finally, the effectiveness of the proposed methods is verified through numerical simulations, encompassing both one-dimensional and high-dimensional scenarios.
\end{abstract}


\end{frontmatter}

\section{Introduction}\label{Section:Introduction}
Cyber-physical systems (CPSs) represent sophisticated engineering systems that seamlessly coordinate, regulate, and integrate perception, communication, computation components \cite{Liu_CPS}.
Multi-sensor fusion estimation (MSFE) stands as a pivotal concern within CPSs, enabling precise monitoring through the amalgamation of valid information sourced from diverse sensors \cite{Chen_CPS}.
However, due to the increasing openness of CPSs, MSFE becomes susceptible to a spectrum of cyberattacks, such as denial of service attacks \cite{Tan_DoS}, false data injection attacks \cite{Zhao_FDI}, and eavesdropping attacks \cite{XinhaoYan_Survey_TASE}.
Notice that eavesdropping attacks pose a significant threat by surreptitiously intercepting communications in CPSs, thereby facilitating the leakage of private data and precipitating more severe active cyberattacks.
Specifically, the state estimates can be utilized by the eavesdroppers to identify models and reconstruct states with the intention of inferring sensitive information and even launching other destructive attacks \cite{XinhaoYan_Survey_TASE}.
Consequently, it is imperative to safeguard the privacy of MSFE to prevent eavesdroppers from accurate state estimation.

In recent years, numerous privacy-preserving methods have been discussed, including cryptography \cite{Darup_HE,Lu_HE}, transmission scheduling \cite{Tsiamis_Code_TAC,Kennedy_Code}, and differential privacy (DP) \cite{Dwork_DP_Book}.
Cryptographic approaches refer to the strategies of utilizing sophisticated secret keys to enhance the security against unauthorized decoding.
Specifically, homomorphic encryption is a cutting-edge technique enables computations to be performed directly on ciphertexts \cite{Darup_HE}.
Transmission scheduling can effectively obscure the original message by strategically introducing packet drops in the communication link \cite{Tsiamis_Code_TAC}, rendering it extremely challenging for eavesdroppers to gather the complete transmitted data.
Traditional DP, also known as centralized differential privacy (CDP), is a classical and well-established privacy-preserving approach stemming from database research \cite{Dwork_DP_Book}.
Its primary objective is to safeguard statistical information through introducing random perturbations.

As is known, cryptographic approaches often necessitate substantial computational resources to ensure high privacy level, while the effectiveness of the transmission scheduling depends on assumptions of the physical links.
By contrast, DP requires minor computation resources and can counter any kinds of eavesdroppers when compared with other encryption methods \cite{XinhaoYan_Survey_TASE}.
Therefore, it has been extensively discussed across numerous fields, such as risk minimization \cite{Liu_DP_Risk}, resilient consensus \cite{Fiore_DP_Consensus,Wang_DP_Consensus}, and linear-quadratic control \cite{Yazdani_DP_Control}.
Regarding state estimation, a considerable number of other estimators were discussed alongside DP, including multi-input multi-output filter \cite{Ny_DP_MIMO}, square root unscented Kalman filter \cite{Yuan_DP}, and set-based estimator \cite{Dawoud_DP}.
For example, an optimized truncated Laplace noise was adopted in \cite{Dawoud_DP}, and this type of noise helped anonymize the centers of the estimated zonotopes in the set-based estimator.

Especially for standard Kalman filtering structure, the differentially private Kalman filter was first proposed in \cite{Ny_DP_TAC}, where the statistical queries about local estimates were perturbed by adding uncertainties, specifically by injecting random noises to either the system inputs or outputs.
This increased the uncertainty associated with the statistical information, thereby enhancing privacy.
For the input mechanism, some parameters were further optimized by solving a semi-definite program in \cite{Degue_DP_TAC}.
Moreover, building upon the above architecture, the MSFE with CDP was explored in \cite{XinhaoYan_DP_TAES,XinhaoYan_DP_Auto}, where fusion center (FC) calculated, aggregated, and transmitted the query outputs about perturbed fusion estimates.
Direct independent noise addition \cite{XinhaoYan_DP_TAES} and two-step sequential noise injection \cite{XinhaoYan_DP_Auto} were respectively proposed for MSFE, and the optimal fusion criteria were both computed in the linear minimum variance sense.
Also, based on the above differentially private distributed fusion estimation structure \cite{XinhaoYan_DP_TAES}, an event-triggered methodology was implemented to minimize energy and computational consumption \cite{Liang_DP}.

However, there exist some inherent limitations to the MSFE when equipped CDP. 
First, the assumptions regarding the sensitivity of CDP is not adequately rational for MSFE.
This is because the sensitivity in CDP is built on the adjacency relation between the original state and the adjacent state, while the latter one only makes sense for large-scale systems, and so does the privacy metric of CDP.
Second, the CDP-based structure introduces an additional communication layer following the FC.
The FC is tasked with broadcasting statistical information to other parties, while such a three-layer architecture is not common for MSFE.
Therefore, to enhance the feasibility and rationality, we consider an recently advanced technique called \emph{local differential privacy (LDP)} for MSFE, and the contribution of this paper can be summarized as follows:
\begin{enumerate}
\item 
It is the first time that LDP is integrated into the MSFE framework to enhance estimation privacy. 
Some rigorous definitions of LDP are proposed, grounded in the intrinsic properties of MSFE, and thus moving away from the irrational assumptions about adjacent states in MSFE with CDP.
Specifically, these include a privacy metric for local state estimates (LSEs) and a sensitivity measure for estimation error covariances.
\item 
Novel differentially private mechanisms without any extra perturbations are proposed for LDP-based MSFE.
Specifically, LDP is proved to be naturally achieved with system intrinsic randomness, provided that specific conditions related to estimation error covariances are met.
Also, Gaussian mechanisms are designed for LDP-based MSFE when the intrinsic condition is not satisfied,
The lower bounds of extra injected noises are designed with the proposed sensitivity and privacy budgets.
\item 
All the mechanisms have been creatively designed to address the specific scenario, where the mathematical expectations for the outputs of the mechanisms are same while the covariances differ.
These approaches represent a significant departure from conventional DP involving different expectations and same covariances.
They offer a fresh perspective and enhanced flexibility for MSFE and even for all unbiased estimation structures.
Finally, leveraging the properties of both intrinsic and external randomness, we propose optimal fusion estimators by minimizing the mean square errors.
\end{enumerate}

The remainder of this paper is organized as follows. 
Section 2 discusses the modelling of the fusion estimator for a CPS.
Meanwhile, the privacy concerns with LDP and the problem formulation are proposed.
In Section 3, we propose some novel definitions about LDP and provide some lemmas to facilitate the subsequent derivation.
Then, one-dimension and high-dimensional intrinsic mechanisms are proposed and the privacy metric is proved to be achieved under certain conditions.
Next, Gaussian mechanisms are designed in Section 4 for both one-dimensional and high-dimensional cases when the above mentioned conditions are not satisfied.
Section 5 considers the simulation on both a one-dimensional oxygen content system and a high-dimensional target tracking system to demonstrate the effectiveness of the proposed methods.
Some intuitive figures about probability distributions and estimation performances are provided.
Finally, the conclusion of this work is given in Section 6.
Besides, the notations that are frequently used throughout the paper are summarized below.

\noindent\textbf{Notations}:
$\mathbb{R}^{n}$ means the set of $n$-dimensional real vectors and $\mathbb{R}^{n\times m}$ means the set of $n\times m$ real matrices.
The notation ``$I_{n}$'' indicates the identity matrix with dimension $n$. Similarly, ``$0_{m \times n}$'' means the zero matrix with dimension $m \times n$.
The superscript ``${\mathrm{T}}$'' is utilized to stand for the transpose of a matrix.
$\mathrm{diag}\{a_{1},\dots,a_{n}\}$ represents a block diagonal matrix and $\mathrm{col}\{a_{1},\dots,a_{n}\}=[a_{1};\dots;a_{n}]=[a_{1}^{\mathrm{T}}\ \dots\ a_{n}^{\mathrm{T}}]^{\mathrm{T}}$ 
represents a column vector, whose elements are ${a_1},\dots,{a_n}$.
${\mathrm{tr}}\left\{\cdot\right\}$ represents the trace of matrix. 
${\mathrm{rank}}\left\{\cdot\right\}$ denotes the rank of matrix. 
$X>(<)0$ means a positive-definite (negative-definite) matrix, while $X\geq(\leq)0$ means a non-negative definite (non-positive definite) matrix.
The notation $|\cdot|$ stands for the absolute value of a scalar, and $\|\cdot\|_{2}$ denotes the $2$-norm of a matrix. 
${\mathbb{E}}\left\{\cdot\right\}$ denotes the mathematical expectation, and ${\mathbb{P}}\left\{\cdot\right\}$ denotes the probability of a random event.
The $n$-dimensional Gaussian distribution for a random vector $X$ with mean $x$ and covariance $P$ is expressed as $\mathcal{N}(X;x,P)=\frac{1}{2\pi^{\frac{n}{2}}|P|^{\frac{1}{2}}}\exp\left(-\frac{1}{2}(X-x)^{\mathrm{T}}P^{-1}(X-x)\right)$.

\section{Problem Formulation} 
Consider a CPS described by the following linear time-invariant state-space model:
\begin{equation}                       
\begin{aligned}
\label{StateSpace}
    &x_{k+1}=Ax_{k}+Bw_{k+1}, \\
    &y_{i,k}=C_{i}x_{k}+D_{i}v_{i,k}\ (i=1,\cdots,L),
\end{aligned}
\end{equation}
where $x_{k}\in\mathbb{R}^{n_{x}}$ denotes the system state at time $k\in\mathbb{Z}_{+}$ and $y_{i,k}\in\mathbb{R}^{n_{y_{i}}}$ represents the measurement of the $i$-th sensor.
Matrices $A\in\mathbb{R}^{n_{x}\times n_{x}}$, $B\in\mathbb{R}^{n_{x}\times n_{w}}$, $C_{i}\in\mathbb{R}^{n_{y_{i}}\times n_{x}}$, and $D_{i}\in\mathbb{R}^{n_{y_{i}}\times n_{v_{i}}}$ are time invariant.
The system noise $w_{k}\in\mathbb{R}^{n_{w}}$ and measurement noise $v_{i,k}\in\mathbb{R}^{n_{v_{i}}}$ are Gaussian distributed and mutually independent, satisfying
\begin{equation}                        
\begin{aligned}
\label{Noise}
    &\mathbb{E}\left\{[w_{k_{1}}^{\mathrm{T}}\  v_{i,k_{1}}^{\mathrm{T}}]^{\mathrm{T}}[w_{k_{2}}^{\mathrm {T}}\ v_{j,k_{2}}^{\mathrm {T}}]\right\}
    = \\
    &\qquad\qquad
    \delta(k_{1},k_{2})\mathrm{diag}\{Q_{w},\delta(i,j)Q_{v_{i}}\},
\end{aligned}
\end{equation}
where $Q_{w}$ and $Q_{v_{i}}\ (\forall\ i)$ are non-negative definite and stand for the respective covariances of $w_{k}$ and $v_{i,k}\ (\forall\ i)$.
Here, $\delta(k_{1},k_{2})$ is an indicator function such that $\delta(k_{1},k_{2})=1$ if $k_{1}=k_{2}$; otherwise, $\delta(k_{1},k_{2})=0$.
Also, the above system is assumed to be controllable and observable, i.e., $\mathrm{rank}\{[B\ AB\ \cdots\ A^{n_{x}-1}B]\}=n_{x}$ and $\mathrm{rank}\{[C^{\mathrm{T}}\ (CA)^{\mathrm{T}}\ \cdots\ (CA^{n_{x}-1})^{\mathrm{T}}]^{\mathrm{T}}\}=n_{x}$.

In this paper, the distributed fusion estimation structure is considered, where the sensors should calculate and transmit LSEs instead of raw measurements.
This is because LDP requires the comparison of consistent local information, which means LSEs are more suitable than measurements with different dimensions.
Specifically, the LSE at each sensor is calculated by the following Kalman filter:
\begin{equation}                      
\label{hat_x_i}
    \hat{x}_{i,k}=A\hat{x}_{i,k-1}+K_{i,k}(y_{i,k}-C_{i}A\hat{x}_{i,k-1}),
\end{equation}
where
\begin{equation}                      
\begin{aligned}
\label{K_i}
    &K_{i,k}=P_{ii,k|k-1}C_{i}^{\mathrm{T}}(C_{i}P_{ii,k|k-1}C_{i}^{\mathrm{T}}+\overline{Q}_{v_{i}})^{-1}, \\
    &P_{ii,k|k-1}=AP_{ii,k-1}A^{\mathrm{T}}+\overline{Q}_{w}, \\
    &P_{ii,k}=(I_{n_{x}}-K_{i,k}C_{i})P_{ii,k|k-1},
\end{aligned}
\end{equation}
with $\overline{Q}_{w}=BQ_{w}B^{\mathrm{T}}$ and $\overline{Q}_{v_{i}}=D_{i}Q_{v_{i}}D_{i}^{\mathrm{T}}$.
Since the Kalman filter converges to the steady state quickly, the steady-state Kalman is directly used for convenience in this paper \cite{XinhaoYan_DP_Auto}.
The steady-state prediction error covariance $\overline{P}_{ii}^{-}\triangleq\lim\limits_{t\to\infty} P_{ii,k|k-1}$ can be computed by resorting to the following discrete Riccati equation:
\begin{equation}                      
\begin{aligned}
\label{P_ii^-}
    \overline{P}_{ii}^{-}=&A\overline{P}_{ii}^{-}A^{\mathrm{T}}+\overline{Q}_{w}-A\overline{P}_{ii}^{-}C_{i}^{\mathrm{T}} \\
    &\times(C_{i}\overline{P}_{ii}^{-}C_{i}^{\mathrm{T}}+\overline{Q}_{v_{i}})^{-1}C_{i}\overline{P}_{ii}^{-}A^{\mathrm{T}}.
\end{aligned}
\end{equation}
In this case, the steady-state estimation covariance and steady-state Kalman gain can be obtained as
\begin{equation}                      
\begin{aligned}
\label{P_ii}
    &\overline{K}_{i}=\overline{P}_{ii}^{-}C_{i}^{\mathrm{T}}(C_{i}\overline{P}_{ii}^{-}C_{i}^{\mathrm{T}}+\overline{Q}_{v_{i}})^{-1}, \\
    &\overline{P}_{ii}=(I_{n_{x}}-\overline{K}_{i}C_{i})\overline{P}_{ii}^{-}.
\end{aligned}
\end{equation}

Besides, due to possible relations among sensors, the steady-state cross estimation error covariance between $i$-th and $j$-th LSEs is also required to be computed:
\begin{eqnarray}                       
\begin{aligned} 
\label{P_ij}
    \overline{P}_{ij}=&(I_{n_{x}}-\overline{K}_{i}C_{i})(A\overline{P}_{ij}A^{\mathrm{T}}+\overline{Q}_{w})(I_{n_{x}}-\overline{K}_{j}C_{j})^{\mathrm{T}}.
\end{aligned}
\end{eqnarray}
Then, after gathering all the LSEs, the distributed fusion estimate (DFE) can be computed at FC by means of weighted summation:
\begin{equation}                       
\label{hat_x_f}
    \hat{x}_{f,k}=\sum_{i=1}^{L}W_{i}\hat{x}_{i,k},
\end{equation}
where the optimal weight in the linear minimum variance sense can be calculated by \cite{XinhaoYan_DP_Auto}:
\begin{equation}                       
\begin{aligned}
\label{W}
    &W\triangleq[W_{1}\ \cdots\ W_{L}]=(I_{a}^{\mathrm{T}}\overline{P}^{-1}I_{a})^{-1}I_{a}^{\mathrm{T}}\overline{P}^{-1}, \\
    &I_{a}\triangleq [I_{n_{x}}^{{\mathrm{T}}}\cdots I_{n_{x}}^{\mathrm{T}}]^{\mathrm{T}}\in \mathbb{R}^{n_{x}L\times n_{x}},
    \overline{P}\triangleq (\overline{P}_{ij})_{n_{x}L\times n_{x}L}.
\end{aligned}
\end{equation}


Due to the potential eavesdroppers in the CPS, we introduce the technique called LDP to protect estimation privacy in the above MSFE system.
In the area of database, the original privacy metric of LDP can be described as $\mathbb{P}(M(d_{i})\in S)\leq e^{\varepsilon}\mathbb{P}(M(d_{j})\in S)$ \cite{Dwork_DP_Book}.
This metric depicts the relation between two data elements under certain privacy budgets.
Corresponding to such a definition, the specific metric for MSFE can be constructed with two LSEs, and the detailed definition is elaborated as follows.

\begin{definition}[($\varepsilon,\delta$)-Local Differential Privacy]
Given privacy budgets $\varepsilon>0$ and $0<\delta<1$, a perturbation mechanism $M:\mathbb{R}^{n_{x}}\to \mathbb{R}^{n_{x}}$ preserves $(\varepsilon,\delta)$-LDP if for all $X_{k}\in\mathrm{Range}\{M\}$ and for all $\hat{x}_{i,k},\ \hat{x}_{j,k}\in\mathbb{R}^{n_{x}}$, we have
\begin{equation}                   
\begin{aligned}
\label{LDP}
    \mathbb{P}(M(\hat{x}_{i,k})=X_{k})\leq e^{\varepsilon}\mathbb{P}(M(\hat{x}_{j,k})=X_{k})+\delta.
\end{aligned}
\end{equation}
\end{definition}

In this metric, $M(\cdot)$ is a randomized mechanism that acts on LSEs and $(\varepsilon,\delta)$ are privacy budgets.
The output of the mechanism is defined as perturbed LSE (PLSE), and the general expression is described by
\begin{equation}                    
    \hat{x}_{i,k}^{p}\triangleq M(\hat{x}_{i,k}).
\end{equation}
Based on these PLSEs, the perturbed DFE (PDFE) $\hat{x}_{f,k}^{p}$ in the linear minimum variance sense is required to be calculated in the similar way of \eqref{hat_x_f} as
\begin{equation}                    
\label{hat_x_f_p}
    \hat{x}_{f,k}^{p}=\sum_{i=1}^{L}W_{i}^{p}\hat{x}_{i,k}^{p},
\end{equation}
where $W^{p}\triangleq[W_{1}^{p}\ \cdots\ W_{L}^{p}]=(I_{a}^{\mathrm{T}}(\overline{P}^{p})^{-1}I_{a})^{-1}I_{a}^{\mathrm{T}}(\overline{P}^{p})^{-1}$.
Since the weights satisfy $\sum_{i=1}^{L}W_{i}^{p}=I_{n_{x}}$, the estimation error covariance of PDFE can be expressed as $P_{f,k}^{p}=\sum_{i=1}^{L}\sum_{j=1}^{L}W_{i}^{p}\tilde{x}_{i,k}^{p}(\tilde{x}_{j,k}^{p})^{\mathrm{T}}(W_{j}^{p})^{\mathrm{T}}$, where $\tilde{x}_{i,k}^{p}\triangleq x_{k}-\hat{x}_{i,k}^{p}$ represents the estimation error of PLSE.

Within the framework of MSFE outlined above, one of the most pivotal objectives is to devise effective mechanisms $M(\cdot)$, aimed at achieving the target of LDP in \eqref{LDP}.
Nonetheless, it is imperative to acknowledge that the presence of system noise introduces an element of randomness into the LSEs.
Consequently, the impact of these system intrinsic uncertainties should be thoroughly considered and discussed in the context of realizing LDP.
Moreover, it is evident that extra perturbation mechanisms will inevitably degrade the estimation performance.
Hence, we also need to minimize the trace of the fusion estimation error covariance $P_{f,k}^{p}$ under the predefined privacy budgets $\varepsilon$ and $\delta$. 
The locally differentially private MSFE structure in this paper is illustrated in Fig. \ref{Fig_Structure}.
Meanwhile, the comparison between the proposed structure with LDP and that with CDP is also shown in the figure, and the detailed differences are discussed in Remark \ref{remark_comparison}.

\begin{figure}[t]         
    \centering
    \includegraphics[width=\columnwidth]{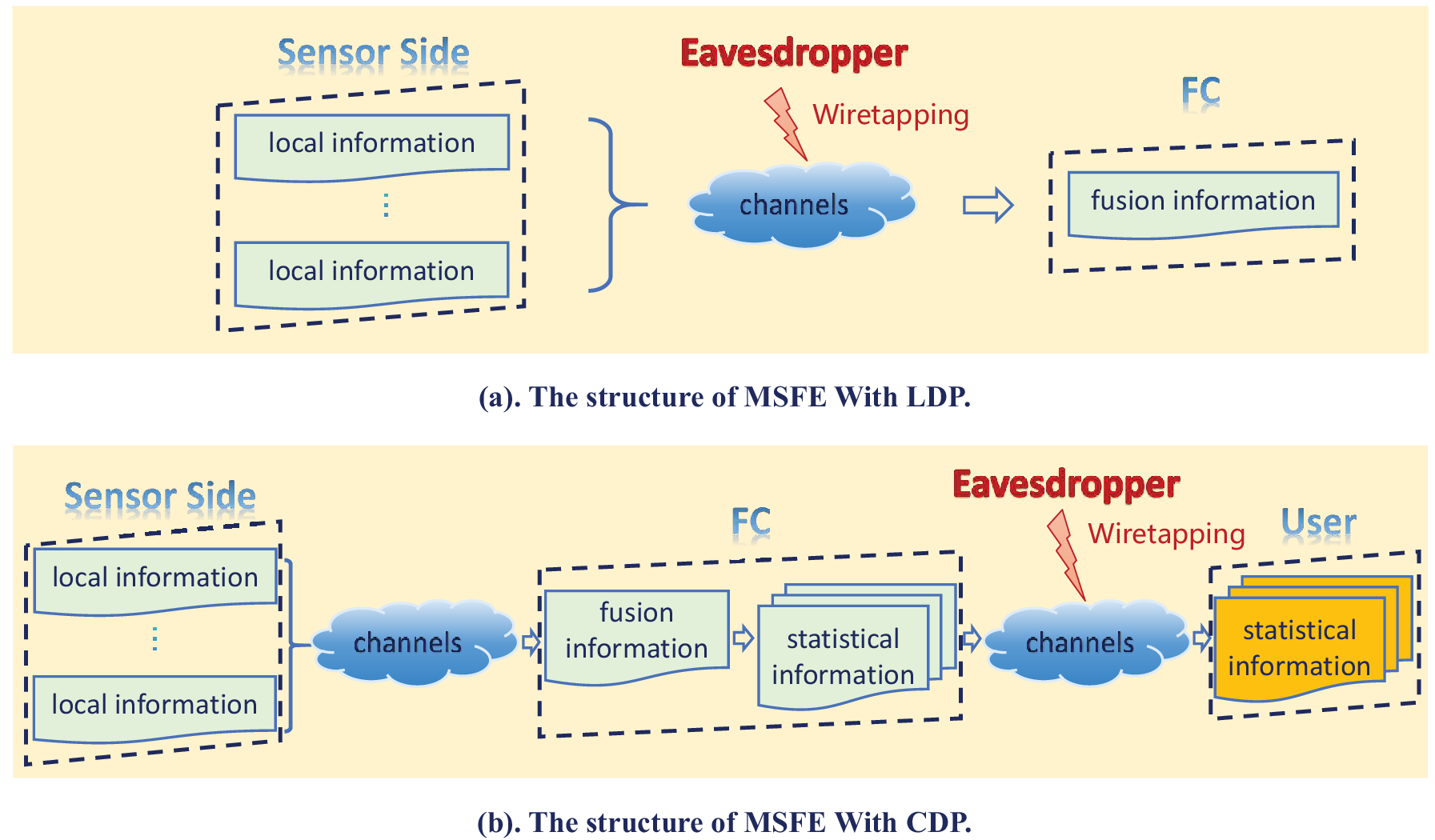}
    \caption{The comparison for MSFEs with LDP and CDP.}
\label{Fig_Structure}
\end{figure}



\begin{remark}
\label{remark_same}
In the realm of privacy-preserving estimation for CPSs, FCs or users typically possesses more information than potential eavesdroppers, including fusion criteria \cite{XinhaoYan_DP_Auto} and acknowledgements \cite{Kennedy_Code}.
However, it is crucial to acknowledge that in certain specific scenarios, eavesdroppers may be able to masquerade as legitimate users. 
In such cases, the information accessible to both the user and the eavesdropper becomes identical, rendering traditional encryption methods utilizing secret keys ineffective.
By contrast, the perturbation mechanism of LDP introduces true-random noises, rather than relying on pseudo-random numbers or other synchronized keys.
This ensures that both users and eavesdroppers receive perturbed data, thereby making it difficult for all the parties to discern original data. 
Consequently, the MSFE combined with LDP provides robust protection against various types of eavesdroppers, significantly elevating the privacy level.
\end{remark}



\begin{remark}
\label{remark_comparison}
As previously discussed, the Kalman filtering methods incorporating CDP have been introduced in \cite{Ny_DP_TAC,Degue_DP_TAC,XinhaoYan_DP_TAES,XinhaoYan_DP_Auto}.
Notably, the input or local perturbation of CDP shares certain similarities with the proposed LDP structure in this paper.
Hence, the following detailed discussion is provided to elucidate the distinctions between them:
\begin{enumerate}
\item \emph{The system dynamics are different.}
The estimation structure with CDP is tailored for large-scale systems, where the statistical information is essential.
Concretely, the fusion estimator or aggregator with CDP is responsible for calculating or transmitting the statistical data, for example, $z_{k}=\sum_{i=1}^{m}\hat{x}_{i,k}$.
Such a design of aggregation aligns with the objective of CDP, which is to protect the individual in statistical information.
By contrast, LDP can be suitable for all multi-party systems, because it does not require statistical queries.
\item \emph{The communication structures are different.}
Generally, the estimation system with CDP comprises three communication layers, where fusion estimator or aggregator will broadcast statistical data to additional parties, and eavesdroppers have the opportunity to intercept this transmission and infer system states.
This layer is typical in database systems while not for general MSFE systems in which the fusion information is exclusively stored at FC and will not be broadcast.
Hence, the perturbation mechanism on LSEs with LDP is more suitable for MSFE systems.
\end{enumerate}
\end{remark}

\section{Intrinsic Mechanism for LDP}
To realize the objective of LDP as defined in \eqref{LDP}, it is necessary to establish the sensitivity between two sensors at the outset.
Traditional CDP in \cite{Ny_DP_TAC,Degue_DP_TAC,XinhaoYan_DP_TAES,XinhaoYan_DP_Auto} considers the sensitivity between two adjacent estimates, which can be expressed as $S\triangleq\|\hat{x}_{k}-\hat{x}_{k}'\|_{2}$.
However, the assumption for the existence of $\hat{x}_{k}'$ is not applicable in the context of MSFE in this paper.
Therefore, we propose an alternative definition of sensitivity that relies on the estimation error covariances, which are readily accessible in MSFE systems.


\begin{definition}[$l_{2}$-Sensitivity]
\label{definition_sensitivity}
The $l_{2}$-sensitivity for MSFE means the maximum of the $2$-norm between any two estimation error covariances, i.e.,
\begin{equation}                    
\label{Delta_2}
    \Delta_{2}
    \triangleq
    \max_{i,j=1,\cdots,L}\|\Delta_{ij}\|_{2}
    \triangleq
    \max_{i,j=1,\cdots,L}\|\overline{P}_{ii}-\overline{P}_{jj}\|_{2}.
\end{equation}
\end{definition}

In other words, it depicts the largest distance between two estimation error covariances that can be easily obtained from the estimator design.
Furthermore, we define the extreme norms for the covariances by the following form:
\begin{equation}                    
\begin{aligned}
\label{P_min}
    \overline{P}_{\min}
    \triangleq
    \min_{i=1,\cdots,L}\|\overline{P}_{ii}\|_{2},\ 
    \overline{P}_{\max}
    \triangleq
    \max_{i=1,\cdots,L}\|\overline{P}_{ii}\|_{2}.
\end{aligned}
\end{equation}
These two extreme norms will be used in the subsequent mechanism design and they can be quickly computed offline by going through the entire spectrum of available norms.
Besides, by resorting to the reverse triangle inequality for $2$-norm of matrices, the relation between the sensitivity in \eqref{Delta_2} and extremums in \eqref{P_min} can be established as follows:
\begin{equation}                    
\begin{aligned}
\label{Delta_2_P_min}
    \Delta_{2}
    \geq
    \overline{P}_{\max}-\overline{P}_{\min}.
\end{aligned}
\end{equation}

Then, some degree of randomness should be incorporated for constructing specific probability distributions to achieve the desired target of LDP as outlined in \eqref{LDP}.
As previously discussed, there already exist some disturbances in MSFE system, including system and measurement noises, which are reflected in estimation error covariance.
It will be beneficial that we directly leverage these system intrinsic disturbances to realize LDP.
This is because with intrinsic disturbances, no additional perturbations would be needed, and thus such a design will not compromise the estimation performance.
Then, prior to presenting the theorem, we provide several pertinent lemmas as outlined below.


\begin{lemma}[Chebyshev Inequality \cite{Proakis_Chebyshev}]
\label{lemma:chebyshev}
Suppose that $x\in\mathbb{R}^{n}$ is an arbitrary random vector with mean $\mu\in\mathbb{R}^{n}$ and covariance $P\in\mathbb{R}^{n\times n}$.
For any positive number $t$, we have
\begin{equation}                      
\label{chebyshev}
    \mathbb{P}(\|x-\mu\|_{2}^{2}>t^{2})\leq \frac{\mathrm{tr}(P)}{t^{2}}.
\end{equation}
\end{lemma}

\begin{proof}
First, let us recall the Markov inequality.
For an arbitrary non-negative random variable $Y\in\mathbb{R}$ and for any positive number $t>0$, the following inequality holds:
\begin{equation}                   
\begin{aligned}
    \mathbb{P}(Y\geq a)\leq \frac{\mathbb{E}\{Y\}}{a}.
\end{aligned}
\end{equation}
In this case, we replace $Y$ by $\|x-\mu\|_{2}^{2}=(x-\mu)^{\mathrm{T}}(x-\mu)$, and thus the following relation is obtained:
\begin{equation}                   
\begin{aligned}
    \mathbb{E}\{Y\}
    =&
    \mathbb{E}\{(x-\mu)^{\mathrm{T}}(x-\mu)\} \\
    =&
    \mathrm{tr}(\mathbb{E}\{(x-\mu)(x-\mu)^{\mathrm{T}}\})
    =
    \mathrm{tr}(P).
\end{aligned}
\end{equation}
By substituting the above relations into the Markov inequality, the result \eqref{chebyshev} in the lemma can be obtained, and the proof is completed.
\end{proof}


\begin{lemma}[Post-Processing \cite{Dwork_DP_Book}]
\label{lemma:post}
Let $M:\mathbb{R}^{n}\to\mathcal{O}$ be a perturbation mechanism of $(\varepsilon,\delta)$-DP, and let $f:\mathcal{O}\to\mathcal{O}'$ be an arbitrary randomized mapping.
Then, $f\circ M:\mathbb{R}^{n}\to\mathcal{O}'$ is still $(\varepsilon,\delta)$-differentially private.
\end{lemma}

\begin{proof}
Here, we prove the proposition for a deterministic function
$f:\mathcal{O}\to\mathcal{O}'$. 
The result follows because any randomized mapping can be decomposed into a convex combination of deterministic functions, and a convex combination of differentially private mechanisms is differentially private.
Fix any pair of neighboring databases $x_{i}$ and $x_{j}$.
Then, fix any event $S\subseteq\mathcal{O}'$ and let $T=\{r\in\mathcal{O}:f(r)\in S\}$.
In this case, we have:
\begin{equation}                   
\begin{aligned}
\label{postprocessing}
    \mathbb{P}(f(M(x_{i})\in S)
    =&
    \mathbb{P}(M(x_{i})\in T) \\
    \leq&
    e^{\varepsilon}\mathbb{P}(M(x_{j})\in T)+\delta \\
    =&
    e^{\varepsilon}\mathbb{P}(f(M(x_{j}))\in S)+\delta.
\end{aligned}
\end{equation}
This completes the proof.
\end{proof}

\begin{lemma}[Matrix Inverse Identity]
\label{lemma:inverse}
Given matrices $A\in\mathbb{R}^{n\times n}$ and $B\in\mathbb{R}^{n\times n}$, and suppose they are invertible.
Then, the following identity holds:
\begin{equation}                   
\begin{aligned}
\label{inverse}
    A^{-1}-B^{-1}
    =
    A^{-1}(B-A)B^{-1}.
\end{aligned}
\end{equation}
\end{lemma}

\begin{proof}
Since $A$ and $B$ are invertible, we can get some relations: $A^{-1}A=AA^{-1}=I_{n}$ and $B^{-1}B=BB^{-1}=I_{n}$.
In this case, the following relations can be established: 
\begin{equation}                   
\begin{aligned}
    A^{-1}=A^{-1}I_{n}=A^{-1}BB^{-1}, \\
    B^{-1}=I_{n}B^{-1}=A^{-1}AB^{-1}.
\end{aligned}
\end{equation}
By substituting them into $A^{-1}-B^{-1}$, the result can be directly obtained.
This completes the proof.
\end{proof}

Subsequently, the theorem with system intrinsic randomness for the one-dimensional system is presented, which we refer to as the \emph{intrinsic mechanism}.
It is demonstrated that LDP is achieved based on the relation among extreme norm, sensitivity and privacy budgets.

\begin{theorem}
\label{theorem_IM}
\rm{\textbf{(One-Dimensional Intrinsic Mechanism)}}
When $n_x=1$, for arbitrary privacy budgets $\varepsilon>0$ and $0<\delta<1$, the LSEs in \eqref{hat_x_i} with estimation error covariances in \eqref{P_ii} preserves $(\varepsilon,\delta)$-LDP in \eqref{LDP} when the following condition holds:
\begin{equation}                       
\begin{aligned}
\label{theorem_IM_eq}
    \overline{P}_{\min}>\frac{\Delta_{2}\sqrt{(\delta+1)^{2}+8\varepsilon\delta}}{2\varepsilon\delta}.
\end{aligned}
\end{equation}
\end{theorem}

\begin{proof}
According to the privacy metric of LDP in \eqref{LDP}, we need to construct the following quantity of privacy loss $\mathcal{L}_{I}$ for one-dimensional intrinsic mechanism \cite{Dwork_DP_Book}:
\begin{equation}                       
\begin{aligned}
    \mathcal{L}_{I}
    &=
    \left|\ln\left(\frac{\mathbb{P}(\hat{x}_{i,k}=X_{k})}{\mathbb{P}(\hat{x}_{j,k}=X_{k})}\right)\right|.
\end{aligned}
\end{equation}
Note that $\hat{x}_{i,k}=x_{k}-\tilde{x}_{i,k}$ and the probability distribution function (PDF) of $\tilde{x}_{i,k}$ is $\mathcal{N}(\tilde{x}_{i,k};0,\overline{P}_{ii})$.
By substituting these specific Gaussian PDFs, the above metric can be rewritten as follows:
\begin{equation}                       
\begin{aligned}
\label{L_I}
    \mathcal{L}_{I}
    &=
    \left|\ln\left(\frac{\mathbb{P}(\tilde{x}_{i,k}=X_{k}-x_{k})}{\mathbb{P}(\tilde{x}_{j,k}=X_{k}-x_{k})}\right)\right| \\
    &=
    \left|\ln\left(\frac{\frac{1}{\sqrt{2\pi \overline{P}_{ii}}}\exp\left(\frac{(X_{k}-x_{k})^{2}}{-2\overline{P}_{ii}}\right)}{\frac{1}{\sqrt{2\pi \overline{P}_{jj}}}\exp\left(\frac{(X_{k}-x_{k})^{2}}{-2\overline{P}_{jj}}\right)}\right)\right| \\
    &=
    \left|\ln\sqrt{\frac{\overline{P}_{jj}}{\overline{P}_{ii}}}
    +\frac{(\overline{P}_{jj}-\overline{P}_{ii})(X_{k}-x_{k})^{2}}{-2\overline{P}_{ii}\overline{P}_{jj}}\right|.
\end{aligned}
\end{equation}
Given the demanded privacy budget, this quantity should be constrained such that $\mathcal{L}_{I}\leq\varepsilon$.
Since the desired result cannot be directly derived from the above expression, we need to find the necessary condition by calculating an upper bound for $\mathcal{L}_{I}$ denoted as $\overline{\mathcal{L}}_{I}$.
This is because the target $\mathcal{L}_{I}\leq\varepsilon$ can be achieved if its upper bound $\overline{\mathcal{L}}_{I}$ satisfies
\begin{equation}                       
\begin{aligned}
    \mathcal{L}_{I}\leq\overline{\mathcal{L}}_{I}\leq\varepsilon.
\end{aligned}
\end{equation}
Consequently, we proceed to compute this upper bound as follows.

With subadditivity or triangle inequality for absolute values, the performance loss \eqref{L_I} can be further separated into the following form:
\begin{equation}                       
\begin{aligned}
\label{overline_L_I}
    \mathcal{L}_{I}
    \leq&
    \left|\ln\sqrt{\frac{\overline{P}_{jj}}{\overline{P}_{ii}}}\right|
    +\left|\frac{\overline{P}_{jj}-\overline{P}_{ii}}{-2\overline{P}_{ii}\overline{P}_{jj}}\right|(X_{k}-x_{k})^{2}.
\end{aligned}
\end{equation}
In this case, the loss can be divided into two parts as $\mathcal{L}_{I}\leq\mathcal{L}_{I,1}+\mathcal{L}_{I,2}(X_{k}-x_{k})^{2}$, where $\mathcal{L}_{I,1}=\left|\ln\sqrt{\frac{\overline{P}_{jj}}{\overline{P}_{ii}}}\right|$ and $\mathcal{L}_{I,2}=\left|\frac{\overline{P}_{jj}-\overline{P}_{ii}}{-2\overline{P}_{ii}\overline{P}_{jj}}\right|$.
Then, the problem of calculating $\overline{\mathcal{L}}_{I}$ can be converted to that of calculating respective bounds of $\mathcal{L}_{I,1}$ and $\mathcal{L}_{I,2}$, i.e., finding $\overline{\mathcal{L}_{I,1}}$ and $\overline{\mathcal{L}_{I,2}}$ such that
\begin{equation}                       
\begin{aligned}
\label{overline_L_I1I2}
    &\mathcal{L}_{I}\leq\overline{\mathcal{L}}_{I}=\overline{\mathcal{L}}_{I,1}+\overline{\mathcal{L}}_{I,2}(X_{k}-x_{k})^{2}\leq\varepsilon, \\
    &\mathrm{where}\ \mathcal{L}_{I,1}\leq\overline{\mathcal{L}}_{I,1},\
    \mathcal{L}_{I,2}\leq\overline{\mathcal{L}_{I,2}}.
\end{aligned}
\end{equation}

On the one hand, we discuss the first part $\mathcal{L}_{I,1}$.
Based on the sensitivity in \eqref{Delta_2} in Definition \ref{definition_sensitivity}, the upper bound of $\mathcal{L}_{I,1}$ can be expressed by
\begin{equation}                       
\begin{aligned}
    \mathcal{L}_{I,1}
    =&
    \left|\ln\sqrt{\frac{\overline{P}_{ii}+\Delta_{ij}}{\overline{P}_{ii}}}\right| \\
    =&
    \frac{1}{2}\ln\left(1+\frac{\Delta_{ij}}{\overline{P}_{ii}}\right) \\
    \leq&
    \frac{1}{2}\ln\left(1+\frac{\Delta_{2}}{\overline{P}_{ii}}\right).
\end{aligned}
\end{equation}
By resorting to the inequality $\ln(1+x)\leq x\ (x\geq 0)$, the upper bound can be further described by
\begin{equation}                       
\begin{aligned}
    \mathcal{L}_{I,1}
    \leq
    \frac{\Delta_{2}}{2\overline{P}_{ii}}.
\end{aligned}
\end{equation}
Then, with minimum in \eqref{P_min} and the relation $\overline{P}_{ii}\geq \overline{P}_{\min}\ (\forall\ i)$ in \eqref{Delta_2_P_min}, the final upper bound of the first part can be formulated as
\begin{equation}                       
\begin{aligned}
\label{overline_L_I1}
    \mathcal{L}_{I,1}
    \leq
    \frac{\Delta_{2}}{2\overline{P}_{\min}}
    =
    \overline{\mathcal{L}}_{I,1}.
\end{aligned}
\end{equation}
On the other hand, with the similar relations mentioned above, the upper bound of the second part can be deduced as follows:
\begin{equation}                       
\begin{aligned}
\label{overline_L_I2}
    \mathcal{L}_{I,2}
    =&
    \left|\frac{\overline{P}_{jj}-\overline{P}_{ii}}{-2\overline{P}_{ii}\overline{P}_{jj}}\right| \\
    \leq&
    \left|\frac{-\Delta_{ij}}{-2\overline{P}_{ii}\overline{P}_{jj}}\right| \\
    \leq&
    \frac{\Delta_{2}}{2\overline{P}_{ii}\overline{P}_{jj}} \\
    \leq&
    \frac{\Delta_{2}}{2\overline{P}_{\min}^{2}}
    =
    \overline{\mathcal{L}}_{I,2}.
\end{aligned}
\end{equation}
Hence, by combining the results in \eqref{overline_L_I1} and \eqref{overline_L_I2}, the parameters in \eqref{overline_L_I1I2} are acquired.

It is obviously that the Gaussian distributions cannot achieve pure $\varepsilon$-LDP.
Hence, it is required to ensure the privacy loss bounded by $\varepsilon$ with probability at least $(1-\delta)$ for $(\varepsilon,\delta)$-LDP \cite{Dwork_DP_Book}:
\begin{equation}                       
\begin{aligned}
\label{PL_I>delta}
    \mathbb{P}(\overline{\mathcal{L}}_{I}\geq\varepsilon)\leq\delta.
\end{aligned}
\end{equation}
In this case, by substituting the above upper bounds \eqref{overline_L_I1} and \eqref{overline_L_I2} into the requirement \eqref{PL_I>delta}, the following condition can be obtained:
\begin{equation}                       
\begin{aligned}
\label{PX>gamma_I}
    \mathbb{P}\left((X_{k}-x_{k})^{2}\geq \gamma_{I}\right)<\delta,
\end{aligned}
\end{equation}
where $\gamma_{I}=\frac{\varepsilon-\overline{\mathcal{L}}_{I,1}}{\overline{\mathcal{L}}_{I,2}}$.
Here, by simplifying Chebyshev inequality \eqref{lemma:chebyshev} stated in Lemma \ref{lemma:chebyshev} to the one-dimensional case, the probability for the Gaussian random variable $X_{k}$ with mean $x_{k}$ and variance $\overline{P}_{ii}$ in \eqref{PX>gamma_I} can be constrained as
\begin{equation}                       
\begin{aligned}
\label{Px>gamma_I2}
    \mathbb{P}\left((X_{k}-x_{k})^{2}\geq \gamma_{I}\right)
    \leq
    \frac{\overline{P}_{ii}}{\gamma_{I}}\ (\forall\ i).
\end{aligned}
\end{equation}

In order to satisfy the requirement in \eqref{PX>gamma_I}, the condition in \eqref{Px>gamma_I2} should be set to be constant, i.e., be independent of ``$i$''.
Thus, by using the extremums defined in \eqref{P_min} and the relation in \eqref{Delta_2_P_min}, the above bound \eqref{Px>gamma_I2} can be further relaxed to 
\begin{equation}                       
\begin{aligned}
    \frac{\overline{P}_{ii}}{\gamma_{I}}
    \leq
    \frac{\overline{P}_{\max}}{\gamma_{I}}
    \leq
    \frac{\overline{P}_{\min}+\Delta_{2}}{\gamma_{I}}\ (\forall\ i).
\end{aligned}
\end{equation}
In this case, the following necessary condition for the privacy target \eqref{PX>gamma_I} can be easily derived:
\begin{equation}                       
\begin{aligned}
\label{condition_I1}
    \frac{\overline{P}_{\min}+\Delta_{2}}{\gamma_{I}}
    <
    \delta. \\
\end{aligned}
\end{equation}

By respectively substituting the detailed upper bounds $\overline{\mathcal{L}}_{I,1}$ in \eqref{overline_L_I1} and $\overline{\mathcal{L}}_{I,2}$ in \eqref{overline_L_I2} into $\gamma_{I}=\frac{\varepsilon-\overline{\mathcal{L}}_{I,1}}{\overline{\mathcal{L}}_{I,2}}$ in \eqref{condition_I1}, the relation among the minimum norm, sensitivity, and privacy budgets can be established as
\begin{equation}                       
\begin{aligned}
\label{condition_I2}
    \varepsilon-\frac{\Delta_{2}}{2\overline{P}_{\min}}
    >
    \frac{\Delta_{2}}{2\overline{P}_{\min}^{2}}\times\frac{\overline{P}_{\max}}{\delta}. \\
\end{aligned}
\end{equation}
For brevity, we write it in the quadratic form about $\overline{P}_{\min}$:
\begin{equation}                       
\begin{aligned}
\label{condition_I3}
    2\varepsilon\delta \overline{P}_{\min}^{2}-(\delta+1)\Delta_{2} \overline{P}_{\min}-\Delta_{2}^{2}
    >
    0.
\end{aligned}
\end{equation}
Since $\varepsilon$ and $\delta$ are both greater than $0$, we have the relation $\sqrt{(\delta+1)^{2}\Delta_{2}^{2}+8\varepsilon\delta\Delta_{2}^{2}}>\sqrt{(\delta+1)^{2}\Delta_{2}^{2}}=(\delta+1)\Delta_{2}$.
In this case, by solving the quadratic equation \eqref{condition_I3} with one variable $\overline{P}_{\min}$, the solution can be obtained:
\begin{equation}                       
\begin{aligned}
\label{condition_I4}
    \overline{P}_{\min}
    >
    \frac{(\delta+1)\Delta_{2}+\sqrt{(\delta+1)^{2}\Delta_{2}^{2}+8\varepsilon\delta\Delta_{2}^{2}}}{4\varepsilon\delta}.
\end{aligned}
\end{equation}
In order to ensure clarity and conciseness for the final result, we further expand the term $(\delta+1)\Delta_{2}$ into $\sqrt{(\delta+1)^{2}\Delta_{2}^{2}+8\varepsilon\delta\Delta_{2}^{2}}$.
Then, the condition \eqref{theorem_IM_eq} in the theorem is acquired.

Meanwhile, we must guarantee that  $\gamma_{I}$ is nonnegative because it represents the possibility.
Therefore, another condition that should be simultaneously satisfied is given below:
\begin{equation}                       
\begin{aligned}
\label{condition_gamma}
    \overline{P}_{\min}
    \geq
    \frac{\Delta_{2}}{2\varepsilon}.
\end{aligned}
\end{equation}
Note that $\frac{\sqrt{(\delta+1)^{2}+8\varepsilon\delta}}{\delta}>\frac{\sqrt{\delta^{2}}}{{\delta}}=1$.
In this case, we can intuitively obtain the following inclusion relation:
\begin{equation}                       
\begin{aligned}
    &\{\overline{P}_{\min}\in\mathbb{R}:\overline{P}_{\min}>\frac{\Delta_{2}\sqrt{(\delta+1)^{2}+8\varepsilon\delta}}{2\varepsilon\delta}\} \\
    &\subset
    \{\overline{P}_{\min}\in\mathbb{R}:\overline{P}_{\min}\geq\frac{\Delta_{2}}{2\varepsilon}\}.
\end{aligned}
\end{equation}
This means that the result can be directly described as the inequality in \eqref{theorem_IM_eq} by considering all the conditions.

Finally, we need to prove $(\varepsilon,\delta)$-LDP with all the conditions mentioned above.
Let us divide $\mathbb{R}$ into $\mathbb{R}=\mathcal{R}_{1}\cup\mathcal{R}_{2}$, where $\mathcal{R}_{1}=\{\overline{P}_{\min}\in\mathbb{R}: \overline{P}_{\min}\leq\frac{\Delta_{2}\sqrt{(\delta+1)^{2}+8\varepsilon\delta}}{2\varepsilon\delta}\}$ and $\mathcal{R}_{2}=\{\overline{P}_{\min}\in\mathbb{R}:\overline{P}_{\min}>\frac{\Delta_{2}\sqrt{(\delta+1)^{2}+8\varepsilon\delta}}{2\varepsilon\delta}\}$.
According to \eqref{PL_I>delta}, one has $\mathbb{P}(M(\hat{x}_{i,k})=X_{k}|\overline{P}_{\min}\in\mathcal{R}_{2})\leq\delta$.
Fix any vector $X_{k}\subseteq\mathbb{R}$, the following relations can be acquired \cite{Dwork_DP_Book}:
\begin{equation}                       
\begin{aligned}
\label{privacyrelation}
    \mathbb{P}(M(\hat{x}_{i,k})=X_{k})
    =&
    \mathbb{P}(M(\hat{x}_{i,k})=X_{k}|\overline{P}_{\min}\in\mathcal{R}_{1}) \\
    &+
    \mathbb{P}(M(\hat{x}_{i,k})=X_{k}|\overline{P}_{\min}\in\mathcal{R}_{2}) \\
    \leq&
    \mathbb{P}(M(\hat{x}_{i,k})=X_{k}|\overline{P}_{\min}\in\mathcal{R}_{1})+\delta \\
    \leq&
    e^{\varepsilon}\mathbb{P}(M(\hat{x}_{j,k})=X_{k}|\overline{P}_{\min}\in\mathcal{R}_{1})+\delta \\
    \leq&
    e^{\varepsilon}\mathbb{P}(M(\hat{x}_{j,k})=X_{k})+\delta.
\end{aligned}
\end{equation}
This yields $(\varepsilon,\delta)$-LDP for the intrinsic mechanism in one-dimension case.
Besides, according to Lemma \ref{lemma:post}, the fusion process \eqref{hat_x_f_p} will not affect the performance of LDP, because it is a post process.
This completes the proof.
\end{proof}



Furthermore, when the dimension of system state is higher than $1$, some of the aforementioned inequalities will cease to be valid.
Consequently, by utilizing some other inequalities and identities for vectors or matrices, the intrinsic mechanism extended for high-dimensional system is provided by the following corollary.

\begin{corollary}
\label{corollary_IM_H}
\rm{\textbf{(High-Dimensional Intrinsic Mechanism)}}
When $n_x\geq 1$, for arbitrary privacy budgets $0<\varepsilon<1$ and $0<\delta<1$, the LSEs \eqref{hat_x_i} in MSFE with estimation error covariances in \eqref{P_ii} preserves $(\varepsilon,\delta)$-LDP in \eqref{LDP} when the following condition holds:
\begin{equation}                       
\begin{aligned}
\label{corollary_IM_H_eq}
    \overline{P}_{\min}>\frac{\Delta_{2}\sqrt{(\delta+n_{x})^{2}+8n_{x}\varepsilon\delta}}{2\varepsilon\delta}.
\end{aligned}
\end{equation}
\end{corollary}

\begin{proof}
Similar to Theorem \ref{theorem_IM}, we first need to construct the quantity $\mathcal{L}_{I}^{H}$ of high-dimensional privacy loss in \eqref{L_I_H}.

\begin{strip}
\vspace{0.1cm}
\hrulefill
\centering
\begin{equation}
\begin{aligned}
\label{L_I_H}
    \mathcal{L}_{I}^{H}
    =&
    \left|\ln\left(\frac{\mathbb{P}(\hat{x}_{i,k}=X_{k})}{\mathbb{P}(\hat{x}_{j,k}=X_{k})}\right)\right| \\
    =&
    \left|\ln\left(\frac{\frac{1}{2\pi^{\frac{n_{x}}{2}}|\overline{P}_{ii}|^{\frac{1}{2}}}\exp(-\frac{1}{2}(X_{k}-x_{k})^{\mathrm{T}}\overline{P}_{ii}^{-1}(X_{k}-x_{k}))}
    {\frac{1}{2\pi^{\frac{n_{x}}{2}}|\overline{P}_{jj}|^{\frac{1}{2}}}\exp(-\frac{1}{2}(X_{k}-x_{k})^{\mathrm{T}}\overline{P}_{jj}^{-1}(X_{k}-x_{k}))}\right)\right| \\
    =&
    \left|\ln\frac{|\overline{P}_{jj}|^{\frac{1}{2}}}{|\overline{P}_{ii}|^{\frac{1}{2}}}-\frac{1}{2}(X_{k}-x_{k})^{\mathrm{T}}(\overline{P}_{ii}^{-1}-\overline{P}_{jj}^{-1})(X_{k}-x_{k})\right| \\
\end{aligned}
\end{equation}
\hrulefill
\vspace{0.1cm}
\end{strip}
By applying the triangle inequality, we decompose it into the following form with two additive parts:
\begin{equation}                       
\begin{aligned}
\label{overline_L_H}
    \mathcal{L}_{I}^{H}
    \leq&
    \left|\ln\frac{|\overline{P}_{jj}|^{\frac{1}{2}}}{|\overline{P}_{ii}|^{\frac{1}{2}}}\right| \\
    &+\left|\frac{1}{2}(X_{k}-x_{k})^{\mathrm{T}}(\overline{P}_{ii}^{-1}-\overline{P}_{jj}^{-1})(X_{k}-x_{k})\right|.
\end{aligned}
\end{equation}
Let us respectively define two parts as $\mathcal{L}_{I,1}^{H}=\left|\ln\frac{|\overline{P}_{jj}|^{\frac{1}{2}}}{|\overline{P}_{ii}|^{\frac{1}{2}}}\right|$ and $\mathcal{L}_{I,2}^{H,x}=\left|\frac{1}{2}(X_{k}-x_{k})^{\mathrm{T}}(\overline{P}_{ii}^{-1}-\overline{P}_{jj}^{-1})(X_{k}-x_{k})\right|$.
Here, we also need to deduce the respective upper bounds.
Nevertheless, unlike the one-dimensional expression in \eqref{L_I}, the square of variable $X_{k}$ cannot be readily extracted from the above quadratic form, thereby significantly increasing the complexity.

According to the definition of determinant and $2-$norm, we can derive the following relationship between them: $|A|=\prod_{i=1}^{n_{x}}\lambda_{i}\leq\lambda_{\max}^{n_{x}}=\|A\|_{2}^{2n_{x}}$.
On the basis of this inequality, the first part can be constrained by the following norm-based expression:
\begin{equation}                       
\begin{aligned}
\label{L_I1_H_ln}
    \mathcal{L}_{I,1}^{H}
    \leq&
    \left|\frac{1}{2}\ln\frac{\|\overline{P}_{ij}\|_{2}^{n_{x}}}{\|\overline{P}_{ii}\|_{2}^{n_{x}}}\right|.
\end{aligned}
\end{equation}
Combining triangle inequality and sensitivity in \eqref{Delta_2}, the above bound can be further described by
\begin{equation}                       
\begin{aligned}
\label{L_I1_H}
    \mathcal{L}_{I,1}^{H}
    \leq&
    \frac{\Delta_{2}^{n_{x}}}{2\overline{P}_{\min}^{n_{x}}}
    =
    \overline{\mathcal{L}}_{I,1}^{H}.
\end{aligned}
\end{equation}

Then, since a quadratic form has the following characteristic: $|x^{T}Ax|\leq \|A\|_{2}\|x\|_{2}^{2}$, we can extract the square of the random variable $(X_{k}-x_{k})\in\mathbb{R}^{n_{x}}$ from the second part as
\begin{equation}                       
\begin{aligned}
\label{overline_L_I2_Hx}
    \mathcal{L}_{I,2}^{H,x}
    \leq&
    \frac{1}{2}\|\overline{P}_{ii}^{-1}-\overline{P}_{jj}^{-1}\|_{2}\|X_{k}-x_{k}\|_{2}^{2}.
\end{aligned}
\end{equation}
Since $\overline{P}_{ii}$ and $\overline{P}_{jj}$ are invertible, we can rewrite the above inequality by resorting to the result \eqref{inverse} of Lemma \ref{lemma:inverse} into the following form:
\begin{equation}                       
\begin{aligned}
\label{overline_L_I2_Hx2}
    \mathcal{L}_{I,2}^{H,x}
    \leq&
    \frac{1}{2}\|\overline{P}_{ii}^{-1}(\overline{P}_{jj}-\overline{P}_{ii})\overline{P}_{jj}^{-1}\|_{2}\|X_{k}-x_{k}\|_{2}^{2}.
\end{aligned}
\end{equation}
Further, by applying submultiplicativity or triangle inequality of the $2$-norm for matrices, one has
\begin{equation}                       
\begin{aligned}
\label{overline_L_I2_Hx3}
    \mathcal{L}_{I,2}^{H,x}
    \leq&
    \frac{1}{2}\|\overline{P}_{ii}^{-1}\|_{2}\|-\Delta_{ij}\|_{2}\|\overline{P}_{jj}^{-1}\|_{2}\|X_{k}-x_{k}\|_{2}^{2}.
\end{aligned}
\end{equation}
Meanwhile, since $\overline{P}_{ii}$ is symmetric positive definite, the $2$-norm of the inverse of $\overline{P}_{ii}$ can be bounded as $\|\overline{P}_{ii}^{-1}\|_{2}\leq \frac{1}{\lambda_{\min}(\overline{P}_{ii})}\leq\frac{1}{\overline{P}_{\min}}\ (\forall\ i)$.
Thus, with sensitivity given in \eqref{Delta_2}, the upper bound of the second part can be acquired in the following norm-based form:
\begin{equation}                       
\begin{aligned}
\label{overline_L_I2_Hx4}
    \mathcal{L}_{I,2}^{H,x}
    \leq&
    \frac{\Delta_{2}}{2\overline{P}_{\min}^{2}}\|X_{k}-x_{k}\|_{2}^{2}
    =
    \overline{\mathcal{L}}_{I,2}^{H}\|X_{k}-x_{k}\|_{2}^{2}.
\end{aligned}
\end{equation}

Next, utilizing the privacy budgets specified in \eqref{LDP}, the objective concerning the aforementioned probability distributions can be formulated as follows:
\begin{equation}                       
\begin{aligned}
\label{PX>gamma_I_H}
    \mathbb{P}(\|X_{k}-x_{k}\|_{2}^{2}\geq \gamma_{I}^{H})<\delta,
\end{aligned}
\end{equation}
where $\gamma_{I}^{H}=\frac{\varepsilon-\overline{\mathcal{L}}_{I,1}^{H}}{\overline{\mathcal{L}}_{I,2}^{H}}$.
According to the Chebyshev inequality tailored for high-dimensional Gaussian random vector in Lemma \ref{lemma:chebyshev}, we can derive a probability upper bound for the above event as
\begin{equation}                       
\begin{aligned}
\label{PX>gamma_I_H2}
    \mathbb{P}(\|X_{k}-x_{k}\|_{2}^{2}\geq\gamma_{I}^{H})
    \leq
    \frac{\mathrm{tr}(\overline{P}_{ii})}{\gamma_{I}^{H}}\ (\forall\ i).
\end{aligned}
\end{equation}
Given the following property for the trace of a matrix: $\mathrm{tr}(\overline{P}_{ii})=\sum_{i=1}^{n_{x}}\lambda_{i}\leq n_{x}\overline{P}_{\max}$.
Then, the requirement for guaranteeing \eqref{PX>gamma_I_H} can be described by combining the relation in \eqref{Delta_2_P_min} as follows:
\begin{equation}                       
\begin{aligned}
\label{condition_I_H1}
    \frac{n_{x}(\overline{P}_{\min}+\Delta_{2})}{t}
    <
    \delta.
\end{aligned}
\end{equation}

Meanwhile, since $\gamma_{I}$ is nonnegative and $\varepsilon<1$, we can derive the following relation:
\begin{equation}                       
\begin{aligned}
\label{condition_I_H2}
    \frac{\Delta_{2}^{n_{x}}}{2\overline{P}_{\min}^{n_{x}}}
    \leq
    \frac{\Delta_{2}}{2\overline{P}_{\min}}
    <
    \varepsilon
    <
    1.
\end{aligned}
\end{equation}
Also, based on the above expression, the first constraint is obtained:
\begin{equation}                       
\begin{aligned}
\label{condition_I_H3}
    \overline{P}_{\min}
    \geq
    \frac{\Delta_{2}}{2\varepsilon}.
\end{aligned}
\end{equation}

Similar to one-dimensional case, based on the relation in \eqref{condition_I_H2}, the inequality \eqref{condition_I_H1} can be converted into following quadratic form about $\overline{P}_{\min}$:
\begin{equation}                       
\begin{aligned}
\label{condition_I_H4}
    2\varepsilon\delta \overline{P}_{\min}^{2}-(\delta+n_{x})\Delta_{2} \overline{P}_{\min}-n_{x}\Delta_{2}^{2}
    >
    0.
\end{aligned}
\end{equation}
Obviously, one has $\sqrt{(\delta+n_{x})^{2}+8\varepsilon\delta}>\sqrt{(\delta+n_{x})^{2}}=(\delta+n_{x})$, which yields the following solution:
\begin{equation}                       
\begin{aligned}
\label{condition_I_H5}
    \overline{P}_{\min}
    >
    \frac{(\delta+n_{x})\Delta_{2}+\sqrt{(\delta+n_{x})^{2}\Delta_{2}^{2}+8n_{x}\varepsilon\delta\Delta_{2}^{2}}}{4\varepsilon\delta}.
\end{aligned}
\end{equation}

Finally, the result \eqref{corollary_IM_H_eq} in Corollary \ref{corollary_IM_H} can be acquired by expanding $(\delta+n_{x})\Delta_{2}$ into $\sqrt{(\delta+n_{x})^{2}\Delta_{2}^{2}+8\varepsilon\delta\Delta_{2}^{2}}$.
Meanwhile, because of $\frac{\sqrt{(\delta+n_{x})^{2}+8\varepsilon\delta}}{\delta}>\frac{\sqrt{\delta^{2}}}{{\delta}}=1$, one has
\begin{equation}                       
\begin{aligned}
    &\{\overline{P}_{\min}\in\mathbb{R}:\overline{P}_{\min}>\frac{\Delta_{2}\sqrt{(\delta+n_{x})^{2}+8\varepsilon\delta}}{2\varepsilon\delta}\} \\
    &\subset
    \{\overline{P}_{\min}\in\mathbb{R}:\overline{P}_{\min}\geq\frac{\Delta_{2}}{2\varepsilon}\}.
\end{aligned}
\end{equation}
Hence, the result of merging two sets can be directly described by the inequality \eqref{corollary_IM_H_eq} in the corollary.
After using post processing lemma and privacy relation \eqref{privacyrelation}, the proof is completed.
\end{proof}

For the above differentially private mechanisms with system intrinsic randomness, it is not required to injecting any extra noises.
In this case, we do not need to optimize anything, and the optimal PDFE with intrinsic mechanism is equivalent to the original DFE in \eqref{hat_x_f}.
In other words, the estimation performance will not not suffer any degradation, marking a significant improvement in the intrinsic mechanism proposed in this paper.

\begin{remark}
It is worth noting that the proposed DP mechanism is designed for two random outputs that share the same mathematical expectation but exhibit different covariances.
In contrast, tradition DP focuses on two outputs with distinct expectations, using this difference as a measure of sensitivity.
At the same time, the covariances of the two outputs in traditional DP are typically assumed to be identical, thus simplifying the result derivation.
In this case, there is no doubt that traditional DP neglects the intrinsic randomness in stochastic systems, such as process and measurement noises, which significantly influence the probability distribution of transmitted data.
To address this, the covariance that contains all the system intrinsic randomness is considered for realizing DP in this paper, thus effectively reducing the level of extra perturbation.
Moreover, the aforementioned properties render our proposed mechanisms suitable not only for MSFE but also for all unbiased estimation frameworks.
\end{remark}

\section{Gaussian Mechanism for LDP}
Although the aforementioned intrinsic mechanisms are effective, the requirement \eqref{theorem_IM_eq} may not be satisfied for some scenarios.
In such instances, additional random noises should be further introduced to help realize LDP.
Since the probability distribution of LSE is Gaussian, it is preferable to insert Gaussian noises, because the Gaussian distribution has the property of superposition.

Then, the PLSE can be designed by the following detailed form:
\begin{equation}                       
\label{hat_x_i_p2}
    \hat{x}_{i,k}^{p}=\hat{x}_{i,k}+a_{i,k}.
\end{equation}
Here, $a_{i,k}$ is the extra injected white Gaussian noise with covariance $Q_{a}I_{n_{x}}$, and all the noises are mutually independent, i.e., $\mathbb{E}\{a_{i,k_{1}}a_{j,k_{2}}\}=0\ (\forall\ i\neq j,\ \forall\ k_{1}\neq k_{2})$.
In this case, the estimation error covariance for PLSE can be calculated by $\overline{P}_{ii}^{p}=\overline{P}_{ii}+Q_{a}I_{n_{x}}$.
Based on the above perturbation form, the Gaussian mechanism is designed in the following theorem and the achievement of LDP is also proved.

\begin{theorem}
\label{theorem_GM}
\rm{\textbf{(One-Dimensional Gaussian Mechanism)}}
When $n_x=1$, for arbitrary privacy budgets $\varepsilon>0$ and $0<\delta<1$, the Gaussian mechanism in \eqref{hat_x_i_p2} with parameter
\begin{equation}                       
\begin{aligned}
\label{theorem_GM_eq1}
    Q_{a}
    \geq
    \frac{\zeta\Delta_{2}}{\varepsilon}-\overline{P}_{\min}.
\end{aligned}
\end{equation}
preserves $(\varepsilon,\delta)$-LDP in \eqref{LDP} when the following condition holds:
\begin{equation}                       
\begin{aligned}
\label{theorem_GM_eq2}
    \zeta
    >
    \frac{\sqrt{(\delta+1)^{2}+8\varepsilon\delta}}{2\delta}.
\end{aligned}
\end{equation}
\end{theorem}

\begin{proof}
Initially, we must also establish the metric $\mathcal{L}_{G}$ for privacy loss like intrinsic mechanisms.
Notice that two Gaussian distributions add up to a new Gaussian distribution.
Consequently, the privacy loss associated with the Gaussian mechanism can be expressed in the following form:
\begin{equation}                       
\begin{aligned}
\label{L_G}
    \mathcal{L}_{G}
    =&
    \left|\ln\left(\frac{\mathbb{P}(M(\hat{x}_{i,k})=X_{k})}{\mathbb{P}(M(\hat{x}_{j,k})=X_{k})}\right)\right| \\
    =&
    \left|\ln\left(\frac{\mathbb{P}(\tilde{x}_{i,k}+a_{i,k}=X_{k}-x_{k})}{\mathbb{P}(\tilde{x}_{j,k}+a_{j,k}=X_{k}-x_{k})}\right)\right| \\
    =&
    \left|\ln\left(\frac{\frac{1}{\sqrt{2\pi \overline{P}_{ii}^{p}}}\exp(\frac{(X_{k}-x_{k})^{2}}{-2\overline{P}_{ii}^{p}})}{\frac{1}{\sqrt{2\pi \overline{P}_{jj}^{p}}}\exp(\frac{(X_{k}-x_{k})^{2}}{-2\overline{P}_{jj}^{p}})}\right)\right| \\
    =&
    \left|\ln\sqrt{\frac{\overline{P}_{jj}^{p}}{\overline{P}_{ii}^{p}}}
    +\frac{(\overline{P}_{jj}^{p}-\overline{P}_{ii}^{p})(X_{k}-x_{k})^{2}}{-2(\overline{P}_{ii}^{p})(\overline{P}_{jj}^{p})}\right| \\
    \leq&
    \left|\ln\sqrt{\frac{\overline{P}_{jj}^{p}}{\overline{P}_{ii}^{p}}}\right|
    +\left|\frac{\overline{P}_{jj}^{p}-\overline{P}_{ii}^{p}}{-2\overline{P}_{ii}^{p}\overline{P}_{jj}^{p}}\right|(X_{k}-x_{k})^{2}.
\end{aligned}
\end{equation}
Similarly, we aim to find the upper bounds $\mathcal{L}_{G,1}\leq\overline{\mathcal{L}}_{G,1}$ and $\mathcal{L}_{G,2}\leq\overline{\mathcal{L}}_{G,2}$ where $\mathcal{L}_{G,1}=\left|\ln\sqrt{\frac{\overline{P}_{jj}^{p}}{\overline{P}_{ii}^{p}}}\right|$ and $\mathcal{L}_{G,2}=\left|\frac{\overline{P}_{jj}^{p}-\overline{P}_{ii}^{p}}{-2\overline{P}_{ii}^{p}\overline{P}_{jj}^{p}}\right|$.
Based on the sensitivity \eqref{Delta_2} and extremum \eqref{P_min}, the upper bound of the first part $\mathcal{L}_{G,1}$ can be easily calculated as
\begin{equation}                       
\begin{aligned}
\label{overline_L_G1}
    \mathcal{L}_{G,1}
    \leq&
    \frac{\Delta_{2}}{2(\overline{P}_{\min}+Q_{a})}
    =
    \overline{\mathcal{L}}_{G,1}.
\end{aligned}
\end{equation}
Analogous to the procedure outlined in \eqref{overline_L_I2}, the second part can be bounded as
\begin{equation}                       
\begin{aligned}
\label{overline_L_G2}
    \mathcal{L}_{G,2}
    \leq&
    \frac{\Delta_{2}}{2(\overline{P}_{\min}+Q_{a})^{2}}
    =\overline{\mathcal{L}}_{G,2}.
\end{aligned}
\end{equation}

Then, to ensure privacy loss bounded by $\varepsilon$ with probability at least $1-\delta$, the following probability condition is required:
\begin{equation}                       
\begin{aligned}
\label{PX>gamma_G}
    \mathbb{P}((X_{k}-x_{k})^{2}\geq \gamma_{G})<\delta,
\end{aligned}
\end{equation}
where $\gamma_{G}=\frac{\varepsilon-\overline{\mathcal{L}}_{G,1}}{\overline{\mathcal{L}}_{G,2}}$.
Meanwhile, with Chebyshev inequality in Lemma \ref{lemma:chebyshev}, one has the following upper bound:
\begin{equation}                       
\begin{aligned}
\label{PX>gamma_G2}
    \mathbb{P}((X_{k}-x_{k})^{2}\geq \gamma_{G})
    \leq
    \frac{\overline{P}_{ii}+Q_{a}}{\gamma_{G}}\ (\forall\ i).
\end{aligned}
\end{equation}
Obviously, the bound \eqref{PX>gamma_G2} can be further expanded as $\frac{\overline{P}_{ii}+Q_{a}}{\gamma_{G}}\leq\frac{\overline{P}_{\max}+Q_{a}}{\gamma_{G}}=\frac{\overline{P}_{\min}+\Delta_{2}+Q_{a}}{\gamma_{G}}$.
By incorporating the requirement \eqref{PX>gamma_G}, the following necessary condition can be derived:
\begin{equation}                       
\begin{aligned}
\label{condition_G1}
    \frac{\overline{P}_{\min}+\Delta_{2}+Q_{a}}{\gamma_{G}}
    <
    \delta.
\end{aligned}
\end{equation}
Providing the existing upper bounds $\overline{\mathcal{L}}_{G,1}$ in \eqref{overline_L_G1} and $\overline{\mathcal{L}}_{G,2}$ in \eqref{overline_L_G2}, we can deduce the following quadratic condition by utilizing the design of $Q_{a}$ specified in \eqref{theorem_GM_eq1}:
\begin{equation}                       
\begin{aligned}
\label{condition_G2}
    2\delta\zeta^{2}-(\delta+1)\zeta-\varepsilon
    >
    0.
\end{aligned}
\end{equation}

Since $\varepsilon$ and $\delta$ are both positive, we can get the relation $\sqrt{(\delta+1)^{2}+8\varepsilon\delta}>(\delta+1)$, and thus the solution to the above inequality can be easily computed as
\begin{equation}                       
\begin{aligned}
\label{condition_G3}
    \zeta
    >
    \frac{\sqrt{(\delta+1)^{2}+8\varepsilon\delta}}{2\delta}.
\end{aligned}
\end{equation}
By expanding $(\delta+1)$ into $\sqrt{(\delta+1)^{2}+8\varepsilon\delta}$, the result \eqref{theorem_GM_eq2} is obtained.
Besides, it is required to guarantee the nonnegativity of $\gamma_{G}$ because it represents the possibility, and thus the following condition should be simultaneously satisfied:
\begin{equation}                       
\begin{aligned}
\label{condition_G4}
    \zeta>\frac{1}{2}.
\end{aligned}
\end{equation}
Notice that $\{\zeta\in\mathbb{R}:\zeta>\frac{1}{2}\}\subset\{\zeta\in\mathbb{R}:\zeta>\frac{\sqrt{(\delta+1)^{2}+8\varepsilon\delta}}{2\delta}\}$ since $\frac{\sqrt{(\delta+1)^{2}+8\varepsilon\delta}}{\delta}>\frac{\sqrt{\delta^{2}}}{\delta}=1$.
Hence, the result \eqref{theorem_GM_eq2} in the theorem holds.
The proof will be completed with post processing in Lemma \ref{lemma:post} and privacy relation \eqref{privacyrelation}.
\end{proof}

Moreover, the extended high-dimensional Gaussian mechanism should be discussed.
By integrating the analogous deduction employed in one-dimensional Gaussian mechanism and high-dimensional intrinsic mechanism, we can straightforwardly derive the following corollary.

\begin{corollary}
\label{corollary_GM_H}
\rm{\textbf{(High-Dimensional Gaussian Mechanism)}}
When $n_x\geq 1$, for arbitrary privacy budgets $0<\varepsilon<1$ and $0<\delta<1$, the Gaussian mechanism in \eqref{hat_x_i_p2} with parameter
\begin{equation}                       
\begin{aligned}
\label{corollary_GM_H_eq1}
    Q_{a}
    \geq
    \frac{\zeta^{H}\Delta_{2}}{\varepsilon}-\overline{P}_{\min}.
\end{aligned}
\end{equation}
preserves $(\varepsilon,\delta)$-LDP in \eqref{LDP} when the following condition holds:
\begin{equation}                       
\begin{aligned}
\label{corollary_GM_H_eq2}
    \zeta^{H}
    >
    \frac{\sqrt{(\delta+n_{x})^{2}+8n_{x}\varepsilon\delta}}{2\delta}.
\end{aligned}
\end{equation}
\end{corollary}

\begin{proof}
At start, we compute the quantity of privacy loss $\mathcal{L}_{G}^{H}$ as shown in \eqref{L_G_H}.

\begin{strip}
\hrulefill
\vspace{0.1cm}
\centering
\begin{equation}
\begin{aligned}
\label{L_G_H}
    \mathcal{L}_{G}^{H}
    =&
    \left|\ln\left(\frac{\mathbb{P}(M(\hat{x}_{i,k})=X_{k})}{\mathbb{P}(M(\hat{x}_{j,k})=X_{k})}\right)\right| \\
    =&
    \left|\ln\left(\frac{\frac{1}{2\pi^{\frac{n_{x}}{2}}|\overline{P}_{ii}^{p}|^{\frac{1}{2}}}\exp(-\frac{1}{2}(X_{k}-x_{k})^{\mathrm{T}}(\overline{P}_{ii}^{p})^{-1}(X_{k}-x_{k}))}
    {\frac{1}{2\pi^{\frac{n_{x}}{2}}|\overline{P}_{jj}^{p}|^{\frac{1}{2}}}\exp(-\frac{1}{2}(X_{k}-x_{k})^{\mathrm{T}}(\overline{P}_{jj}^{p})^{-1}(X_{k}-x_{k}))}\right)\right| \\
    \leq&
    \left|\ln\frac{|\overline{P}_{jj}^{p}|^{\frac{1}{2}}}{|\overline{P}_{ii}^{p}|^{\frac{1}{2}}}\right|
    +\left|\frac{1}{2}(X_{k}-x_{k})^{\mathrm{T}}((\overline{P}_{ii}^{p})^{-1}-(\overline{P}_{jj}^{p})^{-1})(X_{k}-x_{k})\right|.
\end{aligned}
\end{equation}
\vspace{0.1cm}
\hrulefill
\end{strip}

After dividing it into two parts $\mathcal{L}_{G,1}^{H}=\left|\ln\frac{|\overline{P}_{jj}^{p}|^{\frac{1}{2}}}{|\overline{P}_{ii}^{p}|^{\frac{1}{2}}}\right|$ and $\mathcal{L}_{G,2}^{H}=\left|\frac{1}{2}(X_{k}-x_{k})^{\mathrm{T}}((\overline{P}_{ii}^{p})^{-1}-(\overline{P}_{jj}^{p})^{-1})(X_{k}-x_{k})\right|$, we can respectively acquire the following bounds with the inequalities used in deduction of intrinsic mechanism:
\begin{equation}                       
\begin{aligned}
\label{overline_L_G1_H}
    \mathcal{L}_{G,1}^{H}
    \leq&
    \frac{\Delta_{2}^{n_{x}}}{2(\overline{P}_{min}+Q_{a})^{n_{x}}}
    =
    \overline{\mathcal{L}}_{G,1}^{H}.
\end{aligned}
\end{equation}
and
\begin{equation}                       
\begin{aligned}
\label{overline_L_G2_H}
    \mathcal{L}_{G,2}^{H}
    \leq&
    \frac{\Delta_{2}}{2(\overline{P}_{\min}+Q_{a})^{2}}\|X_{k}-x_{k}\|_{2}^{2}
    =
    \overline{\mathcal{L}}_{G,2}\|X_{k}-x_{k}\|_{2}^{2}.
\end{aligned}
\end{equation}

To guarantee the probability $(1-\delta)$, the following target is required to be achieved:
\begin{equation}                       
\begin{aligned}
\label{PX>t_G}
    \mathbb{P}(\|X_{k}-x_{k}\|_{2}^{2}\geq\gamma_{G}^{H})
    <
    \delta,
\end{aligned}
\end{equation}
where $\gamma_{G}^{H}=\frac{\varepsilon-\overline{\mathcal{L}}_{G,1}^{H}}{\overline{\mathcal{L}}_{G,2}^{H}}$.
Owing to the nonnegativity of $\gamma_{G}^{H}$, the following relation can be derived:
\begin{equation}                       
\begin{aligned}
\label{condition2_G_H1}
    \frac{\Delta_{2}^{n_{x}}}{2(\overline{P}_{\min}+Q_{a})^{n_{x}}}
    <
    \frac{\Delta_{2}}{2(\overline{P}_{\min}+Q_{a})}
    <
    \varepsilon
    <
    1.
\end{aligned}
\end{equation}
By substituting the design \eqref{corollary_GM_H_eq1} of $Q_{a}$, one condition of $\zeta^{H}$ can be acquired as
\begin{equation}                       
\begin{aligned}
\label{condition2_G_H2}
    \zeta^{H}>\frac{1}{2}.
\end{aligned}
\end{equation}

Moreover, with Chebyshev inequality \eqref{chebyshev} for a random vector, one has
\begin{equation}                       
\begin{aligned}
\label{PX>gamma_G_H1}
    \mathbb{P}(\|X_{k}-x_{k}\|_{2}^{2}\geq\gamma_{G})
    \leq
    \frac{\mathrm{tr}(\overline{P}_{\max}+Q_{a})}{\gamma_{G}}.
\end{aligned}
\end{equation}
Since $\frac{n_{x}(\overline{P}_{\max}+Q_{a})}{\gamma_{G}^{H}}\leq\frac{n_{x}(\overline{P}_{\min}+\Delta_{2}+Q_{a})}{\gamma_{G}^{H}}$, the following requirement can be obtained with \eqref{PX>t_G}:
\begin{equation}                       
\begin{aligned}
\label{condition_G_H1}
    \frac{n_{x}(\overline{P}_{\min}+\Delta_{2}+Q_{a})}{\gamma_{G}^{H}}
    <
    \delta.
\end{aligned}
\end{equation}
With the upper bounds $\overline{\mathcal{L}}_{G,1}^{H}$ in \eqref{overline_L_G1_H} and $\overline{\mathcal{L}}_{G,2}^{H}$ in \eqref{overline_L_G2_H}, the following condition is constructed:
\begin{equation}                       
\begin{aligned}
\label{condition_G_H2}
    2\delta(\zeta^{H})^{2}-(\delta+n_{x})\zeta^{H}-n_{x}\varepsilon
    >
    0.
\end{aligned}
\end{equation}

Finally, thanks to the relation $\sqrt{(\delta+n_{x})^{2}+8n_{x}\varepsilon\delta}>(\delta+n_{x})$, the lower bound of $\zeta^{H}$ can be deduced as
\begin{equation}                       
\begin{aligned}
\label{condition_G_H3}
    \zeta^{H}
    >
    \frac{(\delta+n_{x})+\sqrt{(\delta+n_{x})^{2}+8\varepsilon\delta}}{4\delta}.
\end{aligned}
\end{equation}
By amplifying $(\delta+n_{x})$ into $\sqrt{(\delta+n_{x})^{2}+8\varepsilon\delta}$, the simplified result \eqref{corollary_GM_H_eq2} in the corollary is obtained.
Based on the relation $\{\zeta^{H}\in\mathbb{R}:\zeta^{H}>\frac{1}{2}\}\subset\{\zeta^{H}\in\mathbb{R}:\zeta^{H}>\frac{\sqrt{(\delta+n_{x})^{2}+8n_{x}\varepsilon\delta}}{2\delta}\}$, the result \eqref{theorem_GM_eq2} is proved by combining post processing \eqref{postprocessing} and privacy relation \eqref{privacyrelation}.
This completes the proof.
\end{proof}

Notice that the covariances of the LSEs have undergone changes with the above Gaussian mechanisms.
Hence, the optimal weight in \eqref{W} necessitates recalculation.
Given that only the distributions of LSEs are different, the fusion structure can be retained as \eqref{hat_x_f_p}.
In this case, adhering to the linear minimum variance sense, the optimal weight can be obtained as
\begin{equation}                      
\begin{aligned}
\label{W_p}
    W^{p}=(I_{a}^{\mathrm{T}}(\overline{P}^{p})^{-1}I_{a})^{-1}I_{a}^{\mathrm{T}}(\overline{P}^{p})^{-1},
\end{aligned}
\end{equation}
where $\overline{P}^{p}\triangleq (\overline{P}_{ij}^{p})_{nL\times nL}$ with
\begin{equation}                      
\begin{aligned}
\label{overline_P_p}
    \overline{P}_{ij}^{p}=\overline{P}_{ij}+Q_{a}I_{n_{x}}\ (\forall\ i=j),
    \overline{P}_{ij}^{p}=\overline{P}_{ij}\ (\forall\ i\neq j).
\end{aligned}
\end{equation}

\section{Simulation Results}
In this section, we apply for both one-dimensional and high-dimensional systems to illustrate the effectiveness of the proposed methods.
Performance analyses, especially probability distributions and estimation errors, are provided to intuitively show the results.

\subsection{One-Dimensional Case}

\begin{figure}[t]         
    \centering
    \subfigure[The real oxygen contents and the estimated contents of LSEs and DFE.]{
        \includegraphics[width=\columnwidth]{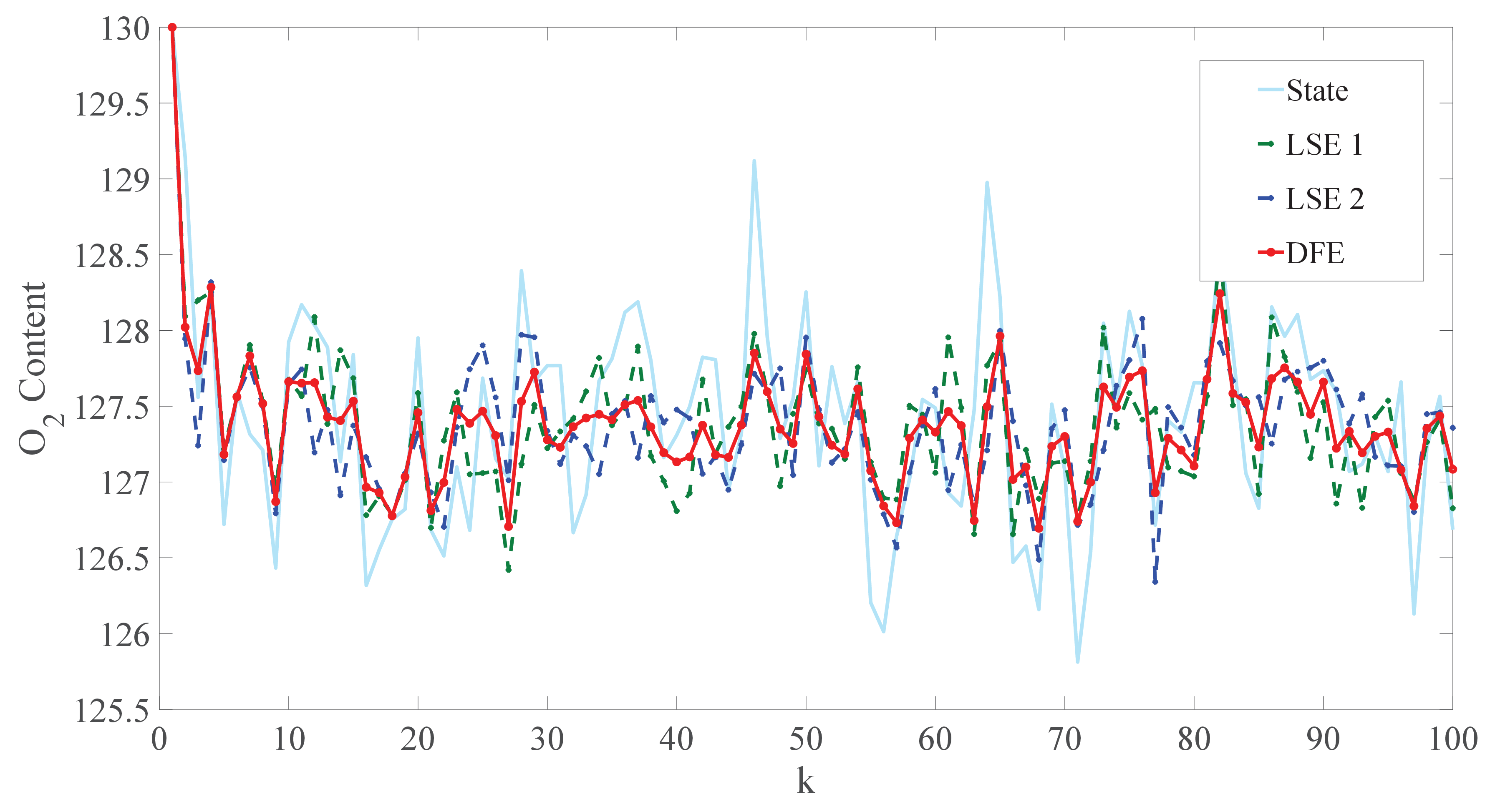}
        \label{Fig_Content_1D}
    }
    \subfigure[The comparison of RMSEs under the intrinsic mechanism.]{
        \includegraphics[width=\columnwidth]{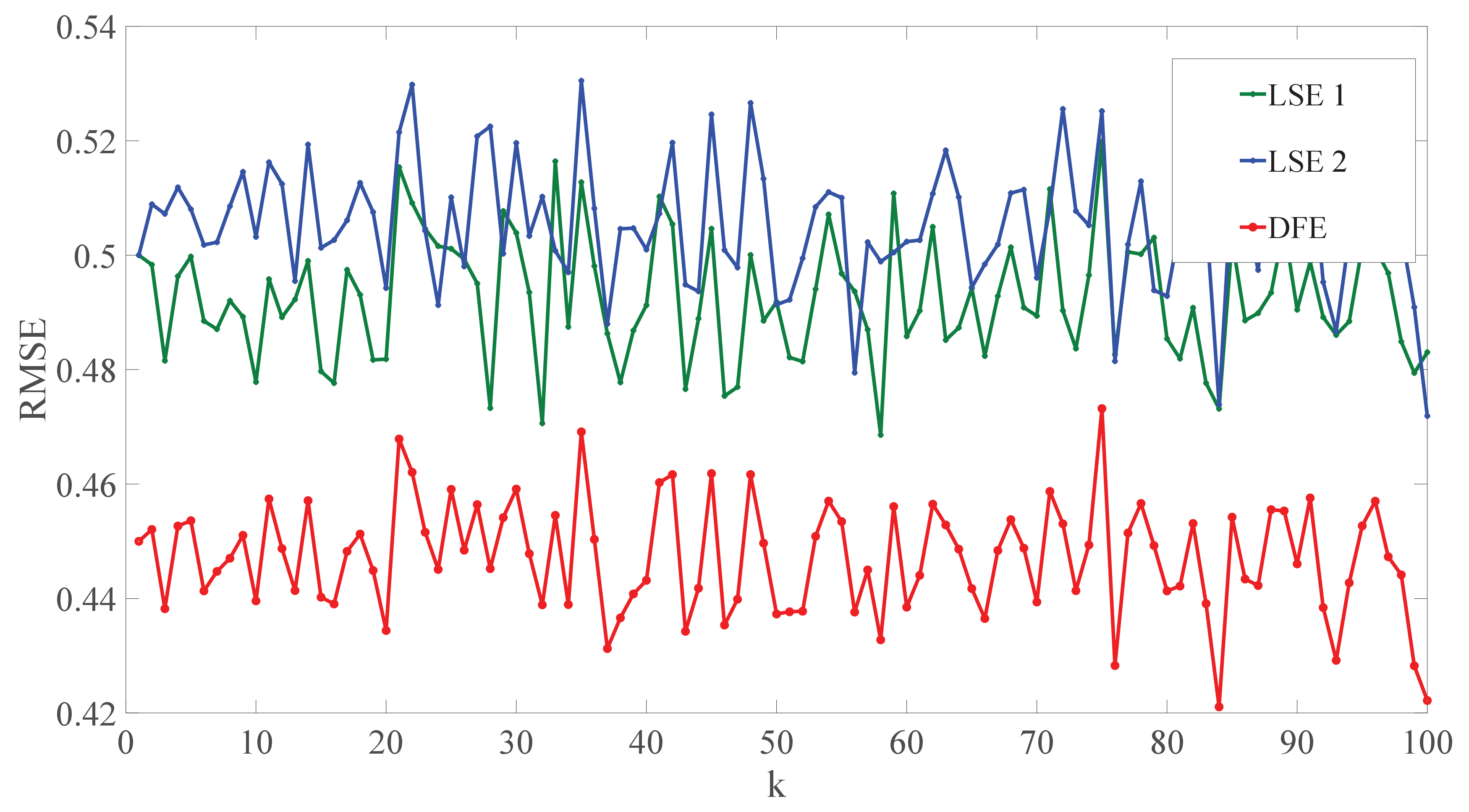}
        \label{Fig_RMSE_1D}
    }
    \caption{The comparison for LSEs and DFE under the intrinsic mechanism.}
\end{figure}

First, let us consider a blood oxygen concentration system for a child during surgery \cite{Ivanov_Oxygen}.
The oxygen content can be described by the following dynamics:
\begin{equation}                      
\begin{aligned}
    a_{k+1}
    =&
    (1-f)(1.34Hb+0.003(c_{1}u_{k})+c_{2}e_{k}) \\
    &+f(a_{k}-\mu)+w_{k},
\end{aligned}
\end{equation}
where $f$ represents the fraction of shunted blood, $c_{1}=P_{ATM}-P_{H_{2}O}$, $u_{k}=(F_{i}O_{2})_{k}$ denotes the fraction of oxygen in inhaled air and is controlled by clinicians, $c_{2}=(1-u_{k}(1-RQ))/RQ$, and $e_{k}=(P_{A}CO_{2})_{k}$ represents the partial pressure of $CO_{2}$ in the alveoli.
Besides, the covariance for noise $w_{k}$ is set as
$Q_{w}=0.4$.
Then, two concentration sensors are deployed to measure the real-time blood oxygen content:
\begin{equation}                      
\begin{aligned}
    y_{1,k}
    =
    a_{k}+v_{1,k},\
    y_{2,k}
    =
    a_{k}+v_{2,k},
\end{aligned}
\end{equation}
where $v_{1,k}$ and $v_{2,k}$ are the measurement noises with the covariance $Q_{v_{1}}=0.6$ and $Q_{v_{2}}=0.7$, respectively.

Here, we define the state as the oxygen content, i.e., $x_{k}\triangleq
a_{k}\in\mathbb{R}$.
Given the parameters $f=0.2$, $Hb=12\mathrm{g/dL}$, $P_{ATM}=760\mathrm{mmHg}$, $P_{H_{2}O}=47\mathrm{mmHg}$, $\mu=5\mathrm{mL/dL}$, $RQ=0.8$.
The, the state-space model in \eqref{StateSpace} can be constructed as
\begin{equation}                      
\begin{aligned}
    &x_{k+1}
    =
    Ax_{k}+w_{k}+U_{k}, \\
    &y_{i,k}
    =
    C_{i}x_{k}+v_{i,k}\ (i=1,2),
\end{aligned}
\end{equation}
where
$U_{k}=(1-f)(1.34Hb+0.003(c_{1}u_{k})+c_{2}e_{k})-f\mu$,
$A=f$,
$B=1$,
$C_{1}=1$,
$C_{2}=1$, and
$D=1$.
With the above model, all the steady-state covariances can be computed and thus we can get the minimum norm in \eqref{P_min} as $\overline{P}_{\min}=0.2435$.
Meanwhile, the privacy budgets of LDP are determined as $\varepsilon=0.8$ and $\delta=0.2$ by the user.
In this case, the following relation can be calculated:
\begin{equation}                       
\begin{aligned}
    \overline{P}_{\min}
    =
    0.2435
    >
    0.0628
    =
    \frac{\Delta_{2}\sqrt{(\delta+1)^{2}+8\varepsilon\delta}}{2\varepsilon\delta},
\end{aligned}
\end{equation}
which satisfies the condition \eqref{theorem_IM_eq} in Theorem \ref{theorem_IM}.
This indicates that the system intrinsic randomness sufficiently contributes to the privacy metric of LDP, and thus obviating the necessity of injecting additional noises.

Due to the intrinsic randomness in CPSs, the estimation performances are assessed by root mean square errors (RMSEs) in this paper, and they are approximated by $1000$ runs of the Monte Carlo method.
The real oxygen contents and the estimated contents of LSEs and DFE are plotted in Fig. \ref{Fig_Content_1D}, where we can see that all the estimates can well track the target state.
Meanwhile, their RMSEs are shown in Fig. \ref{Fig_RMSE_1D}.
It can be realized from this figure that all the estimates are stable because of the bounded RMSE.
Moreover, the fusion estimation performance is better than locals, showing the advancement of the fusion approach.

\begin{figure}[t]         
    \centering
    \includegraphics[width=\columnwidth]{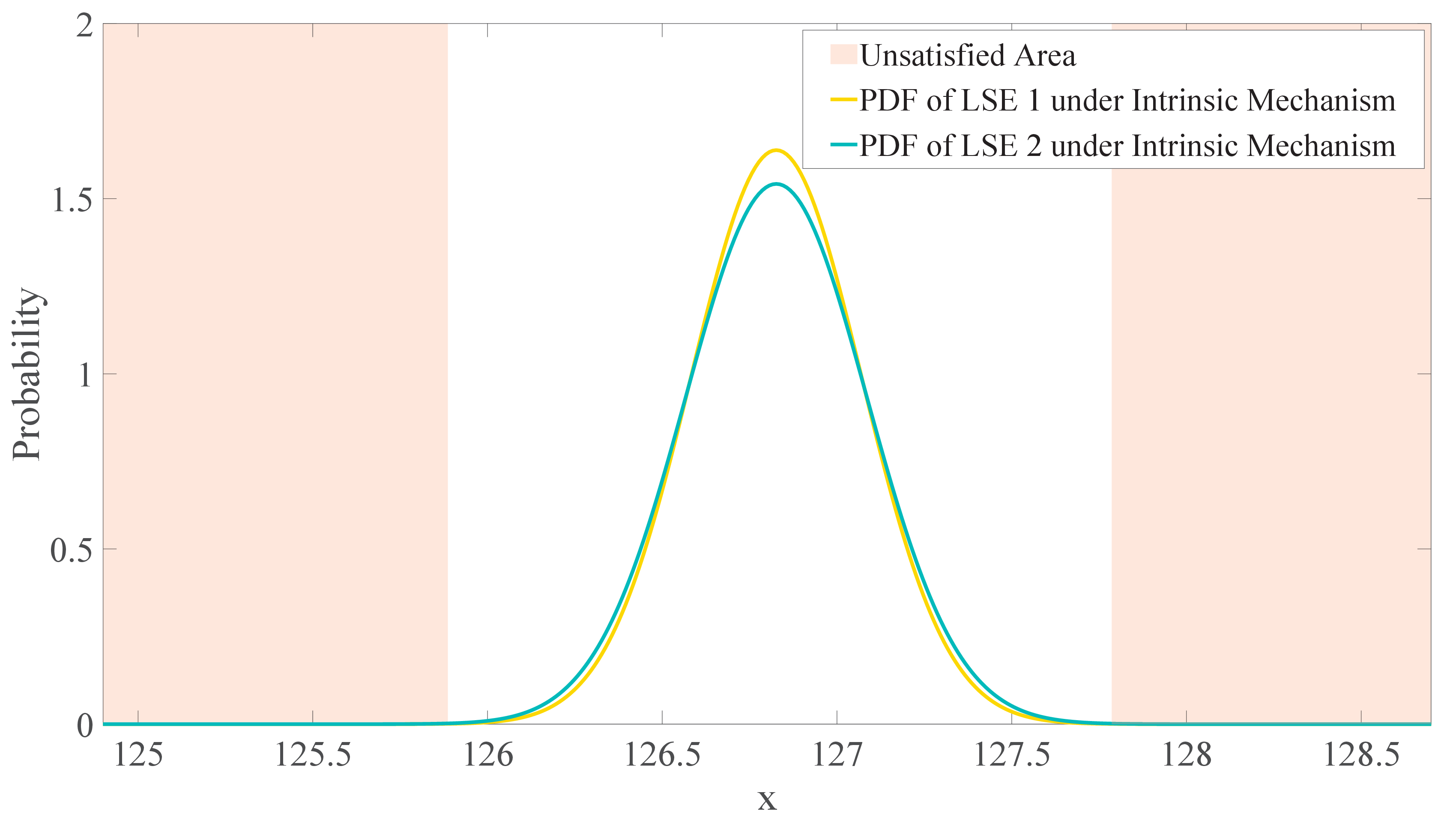}
    \caption{The comparison of the PDFs for two LSEs under the intrinsic mechanism.}
    \label{Fig_PDF_1D}
\end{figure}

Then, based on the proposed intrinsic mechanism, we draw the PDFs of two LSEs in Fig. \ref{Fig_PDF_1D}.
The pink part is the area that the privacy target is not satisfied, denoted as $\mathcal{J}$, which means
\begin{equation}                       
\begin{aligned}
    \mathcal{L}_{I}
    =
    \left|\ln\left(\frac{\mathbb{P}(\hat{x}_{i,k}=X_{k})}{\mathbb{P}(\hat{x}_{j,k}=X_{k})}\right)\right|
    >
    \varepsilon,\ \forall\ X_{k}\in\mathcal{J}.
\end{aligned}
\end{equation}
Intuitively, the unsatisfied region is very small (note that the probability is measured by the total area instead of length).
It is determined by the design of high-level privacy budgets $\varepsilon$ and $\delta$.
By contrast, the target condition $\mathcal{L}_{I}=\left|\ln\left(\frac{\mathbb{P}(\hat{x}_{i,k}=X_{k})}{\mathbb{P}(\hat{x}_{j,k}=X_{k})}\right)\right|\leq\varepsilon$ can be satisfied in the majority of scenario, showing the guarantee of sufficient privacy.

\subsection{High-Dimensional Case}

\begin{figure}[t]         
    \centering
    \subfigure[Original PDFs for LSEs.]{
        \includegraphics[width=\columnwidth]{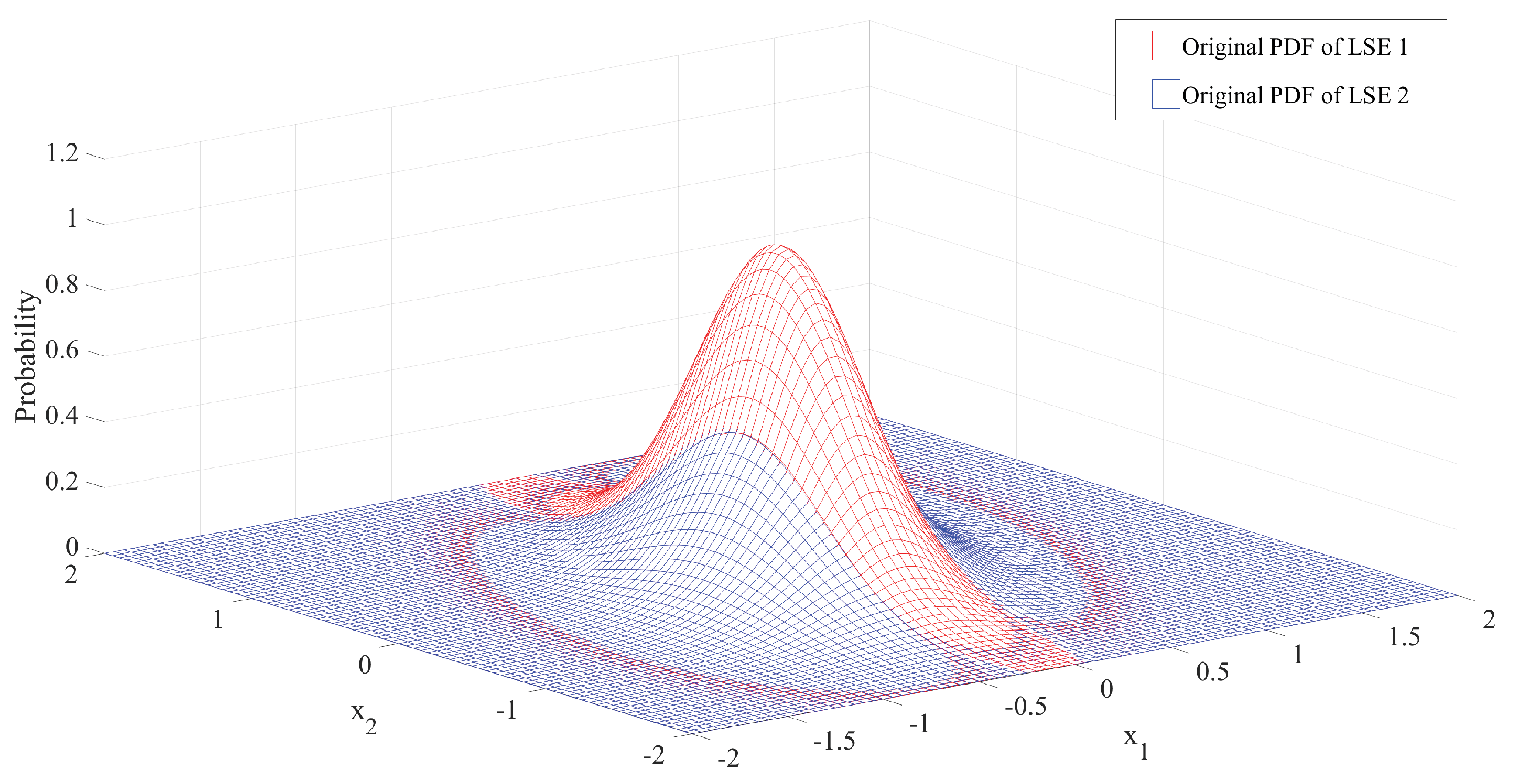}
        \label{Fig_PDF}
    }
    \subfigure[PDFs for PLSEs under Gaussian mechanism.]{
        \includegraphics[width=\columnwidth]{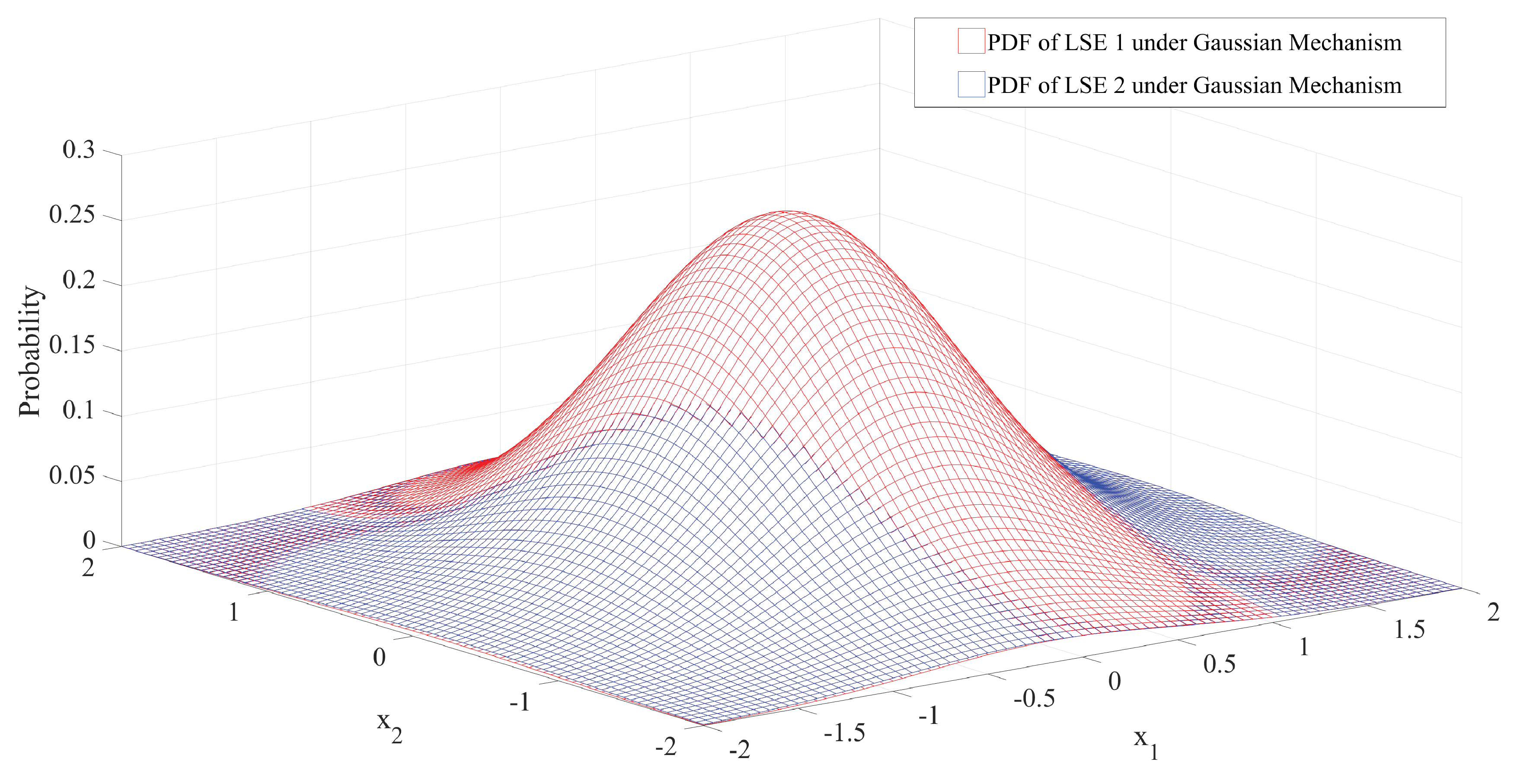}
        \label{Fig_PDF_Gaussian}
    }
    \caption{The comparison of the PDFs for LSEs under the Gaussian mechanism.}
\end{figure}

\begin{figure}[t]         
    \centering
    \includegraphics[width=\columnwidth]{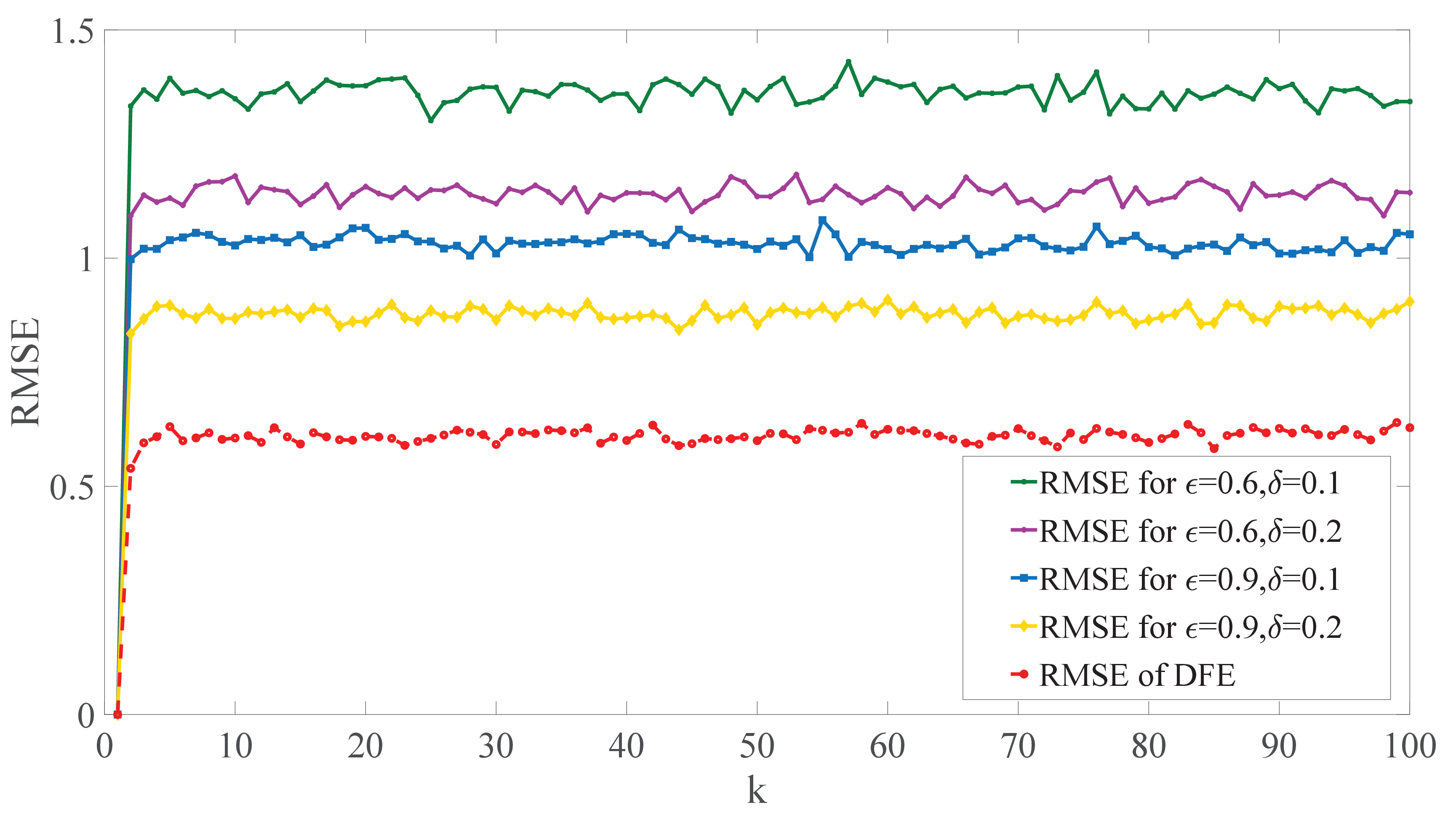}
    \caption{The relations between RMSEs and privacy budgets under the Gaussian mechanism.}
    \label{Fig_RMSE}
\end{figure}

Moreover, a two-dimensional target tracking system observed by two positioning sensors \cite{XinhaoYan_DP_TAES} is proposed to further demonstrate the effectiveness of the high-dimensional Gaussian mechanism.
The state-space model can be established by the following parameters:
$
A=
\begin{bmatrix}
1 & T \\
0 & 1
\end{bmatrix},\ 
B=
\begin{bmatrix}
0.5T^{2} \\
T
\end{bmatrix},\ 
C_{1}=
\begin{bmatrix}
1 & 0
\end{bmatrix},\ 
C_{2}=
\begin{bmatrix}
1 & 0
\end{bmatrix},\ 
D=1
$,
where $T$ is the sampling time that is set as $1$s.
Meanwhile, the covariances of system noises are set as 
$Q_{w}=0.1,\   
Q_{v_{1}}=0.2,\ 
Q_{v_{2}}=0.1$.
Meanwhile, the privacy budgets of LDP are determined as $\varepsilon=0.9$ and $\delta=0.2$.
In this case, the condition of high-dimension intrinsic mechanism in Corollary \ref{corollary_IM_H} is not satisfied, and thus we compute the lower bound of the perturbation as $Q_{a}=0.3950$ based on the high-dimensional Gaussian mechanism in Corollary \ref{corollary_GM_H}.

First, the two-dimensional PDFs are plotted in Fig. \ref{Fig_PDF}.
We can see that the covariances for two PDFs increase with Gaussian mechanism, and meanwhile, they become closer when compared to their original values such that the condition in Corollary \ref{corollary_GM_H} can be satisfied.
Also, the final estimation performances for PDFEs are assessed by resorting to RMSEs as shown in Fig. \ref{Fig_RMSE}.
Furthermore, we discuss the relation between the RMSE and privacy budgets by computing RMSEs under different parameters.
The relations between RMSEs and privacy budgets are also demonstrated in the Fig. \ref{Fig_RMSE}.
Specifically, the detailed cases are $\varepsilon=0.6,\delta=0,1$, $\varepsilon=0.6,\delta=0,2$, $\varepsilon=0.9,\delta=0,1$, and $\varepsilon=0.9,\delta=0,2$.
In this case, we can intuitively realize the relation by observing two of them, i.e., the RMSE decreases with the increase of privacy budgets $\varepsilon$ and $\delta$, while the rate with respect to $\delta$ is a bit larger.

\section{Conclusion}
This paper investigated a privacy-preserving MSFE method leveraging LDP.
We introduced the following two novel definitions about LDP: 
1) The sensitivity incorporating estimation error covariances of LSEs;
2) The privacy metric about the probability distributions related to LSEs.
Then, the achievement of LDP was demonstrated under the condition of system intrinsic randomness, particularly considering estimation error covariances in this paper, encompassing both one-dimensional and high-dimensional scenarios.
Furthermore, we proved the realization of the Gaussian mechanisms when the intrinsic condition was not satisfied, i.e., determined the lower bound for the covariance of additionally injected Gaussian noises.
Moreover, the optimal fusion estimators were designed in the linear minimum variance sense under all the intrinsic and Gaussian mechanisms.
Finally, the effectiveness of the proposed approaches was validated through two numerical simulations, including a one-dimensional oxygen content system and a high-dimensional target tracking system.






\vfill

\end{document}